\documentclass[11pt]{article}
\usepackage[english]{babel}

\usepackage[letterpaper,top=2cm,bottom=1.5cm,left=1in,right=1in,marginparwidth=1.75cm]{geometry}


\usepackage{setspace}
\usepackage{amsmath}
\usepackage{amsthm,thm-restate}
\usepackage{graphicx}
\usepackage{amssymb}
\usepackage{wrapfig}
\usepackage{titlesec}
\usepackage{bbm}
\usepackage[colorlinks=true, allcolors=blue]{hyperref}
\usepackage[capitalise,nameinlink]{cleveref}
\usepackage[textsize=tiny]{todonotes}

\Crefname{algocf}{Algorithm}{Algorithms}
\crefname{algocfline}{line}{lines}
\Crefname{invariant}{Invariant}{Invariants}
\Crefname{claim}{Claim}{Claims}
\Crefname{subclaim}{Subclaim}{Subclaims}

\usepackage[ruled,vlined,linesnumbered,algonl]{algorithm2e}
\SetEndCharOfAlgoLine{}
\SetKwComment{Comment}{\footnotesize$\triangleright$\ }{}

\SetCommentSty{mycommfont}

\usepackage{tikz}

\usepackage{nicefrac}

\usepackage[shortlabels]{enumitem}
\usepackage[font={small,it}]{caption}

\newtheorem{theorem}{Theorem}[section]
\newtheorem{lemma}[theorem]{Lemma}
\newtheorem{definition}[theorem]{Definition}

\usepackage{booktabs}
\usepackage{mathabx}
\usepackage{csquotes}
\usepackage{empheq}
\usepackage{bbm}
\def\P#1{\mathrm{\mathbb{P}}\left[#1\right]}

\usepackage{subcaption}
\usepackage{cleveref}

\newcommand{\Exp}{\mbox{Exp}}

\newcommand{\E}[1]{\mathrm{\bf E}\left[#1\right]}

\newcommand{\cost}{\text{Cost}}
\newcommand{\A}{A}
\newcommand{\costbar}{\widebar{Cost}}
\newcommand{\guess}{\widehat{W}}
\newcommand{\PoisCost}{\Omega}

\newtheorem{customthm}{Theorem}
\newtheorem{customcor}{Corollary}

\author{Zhouzi Li \thanks{zhouzil@andrew.cmu.edu}\and Keerthana Gurushankar\and Mor Harchol-Balter\and Alan Scheller-Wolf \\ Carnegie Mellon University}

\title{Improving Upon the generalized $c\mu$ rule: a Whittle approach}

\begin{document}

\maketitle


\begin{abstract}

Scheduling a stream of jobs whose holding cost changes over time is a classic and practical problem.  Specifically, each job is associated with a holding cost (penalty), where a job's instantaneous holding cost is some increasing function of its class and current age (the time it has spent in the system since its arrival). The goal is to schedule the jobs to minimize the time-average total holding cost across all jobs.


The seminal paper on this problem, by Van Mieghem in 1995 \cite{van1995dynamic}, introduced the generalized $c\mu$ rule for scheduling jobs.  Since then, this problem has attracted significant interest  but remains challenging due to the absence of a finite-dimensional state space formulation. Consequently, subsequent works focus on more tractable versions of this problem.

This paper returns to the original problem, deriving a heuristic that empirically improves upon the generalized $c \mu$ rule and all existing heuristics.
Our approach is to first translate the holding cost minimization problem to a novel Restless Multi-Armed Bandit (R-MAB) problem with a {\em finite} number of arms. Based on our R-MAB, we derive a novel Whittle Index policy, which is both elegant and intuitive. 

\end{abstract}



\section{Introduction}

Since the seminal paper by Van Mieghem in 1995~\cite{van1995dynamic}, the problem of scheduling jobs with Time-Varying Holding Cost (the TVHC problem) has been an important topic in the operations literature: In a single-server multi-class system with $k$ classes of jobs, 
each job incurs a (time-varying) holding cost for every unit of time it remains in the system, and the goal of the scheduling policy is to minimize the time-average total holding cost across all jobs.  Specifically, define the \textit{age} of a job to be the time it has spent in the system since it arrived. For a job of class $i$, let $c_i(t)$ be the instantaneous holding cost when the job's age is $t$. We allow different classes to have different holding-cost functions (see Figure~\ref{fig:different hold}), but assume these functions are non-decreasing. Note that non-decreasing {\em instantaneous holding cost} is equivalent to convex {\em accumulated holding cost}, which is assumed in~\cite{van1995dynamic} and all its follow-on works. Also, throughout this paper, we assume that job sizes are {\em exponentially} distributed unless otherwise specified, and we assume that the job arrival process is Poisson. Let $\lambda_i$ and $\mu_i$, respectively,  denote the arrival rate and the completion rate of class $i$ jobs.

In the special case when the holding cost of each class is \emph{a constant function} (where $c_i(t)=c_i$), the optimal policy is the famous $c\mu$ rule, which always (preemptively) runs the job with the highest product $c_i \cdot \mu_i$~\cite{cox2020queues,Buyukkoc_Varaiya_Walrand_1985}. However, in the general case, this problem is much more complicated and the optimal policy is not known. 
In \cite{van1995dynamic},
 an analogue of the $c\mu$ rule is proposed, which is now commonly referred to as the \emph{generalized $c\mu$ rule}: The priority of a job is given by the product of its instantaneous holding cost, $c_i(t)$, and its instantaneous failure rate $\mu_i(t)$. While the generalized $c\mu$ rule is defined for general job sizes in \cite{van1995dynamic}, if the job sizes are exponentially distributed, the failure rate for class $i$ is just $\mu_i$. We discuss the generalized $c\mu$ rule in Section~\ref{sec:intro gen}.

Note that both the generalized $c \mu$ rule and the $c \mu$ rule are {\em index policies}: Each job has an index (in both cases the value $c_i(t) \mu_i(t)$) which is a function of only the job's state, and the policy always serves that job with the highest index.   Index policies are both simple and powerful.
This paper aims to solve the TVHC problem within the class of preemptive {\em index policies}, yielding a heuristic solution to the problem in general. We now review the literature, and motivate our approach. 



\subsection{The generalized $c\mu$ rule}
\label{sec:intro gen}

The generalized $c\mu$-rule has been shown to be asymptotically optimal for M/G/1 queues in the diffusion limit regime, where both the arrival and service rates go to infinity (and the total load goes to 1)~\cite{van1995dynamic}. Under the same diffusion limit, its analogues have also been shown to be asymptotically optimal for more complicated settings such as the multi-server setting and systems with abandonment (\cite{mandelbaum2004scheduling,atar2010cmu,  long2020dynamic}).
However, outside of the diffusion limit regime, it is known that the generalized $c\mu$-rule can perform poorly. Here we give an example to illustrate this.
 Consider a system with two classes of jobs: one with deadlines and one without. For deadline-based jobs, the holding cost is zero before their deadline but becomes significantly higher once the deadline is passed, while jobs without deadlines have a constant holding cost (see Figure~\ref{fig: example1}).  Assume that both classes have job sizes following the same $\Exp(\mu)$ distribution. Under the generalized $c\mu$ rule, the system grants priority to deadline-based jobs only \emph{after} they have missed their deadlines. Intuitively, however, it might make more sense to prioritize these jobs \emph{before} their deadline is reached. Simulation results show that by prioritizing the deadline-based jobs before their deadline, the overall holding cost can be substantially reduced in normal-traffic scenarios (see Figure~\ref{fig:1-deadline-plot}). 

Thus, the generalized $c\mu$ rule can be highly suboptimal under normal traffic (even with exponential job sizes), underscoring the need for a better scheduling policy. 

\begin{figure}[h]
    \begin{minipage}[c]{0.5\linewidth}
        \centering
        \includegraphics[width=0.8\textwidth]{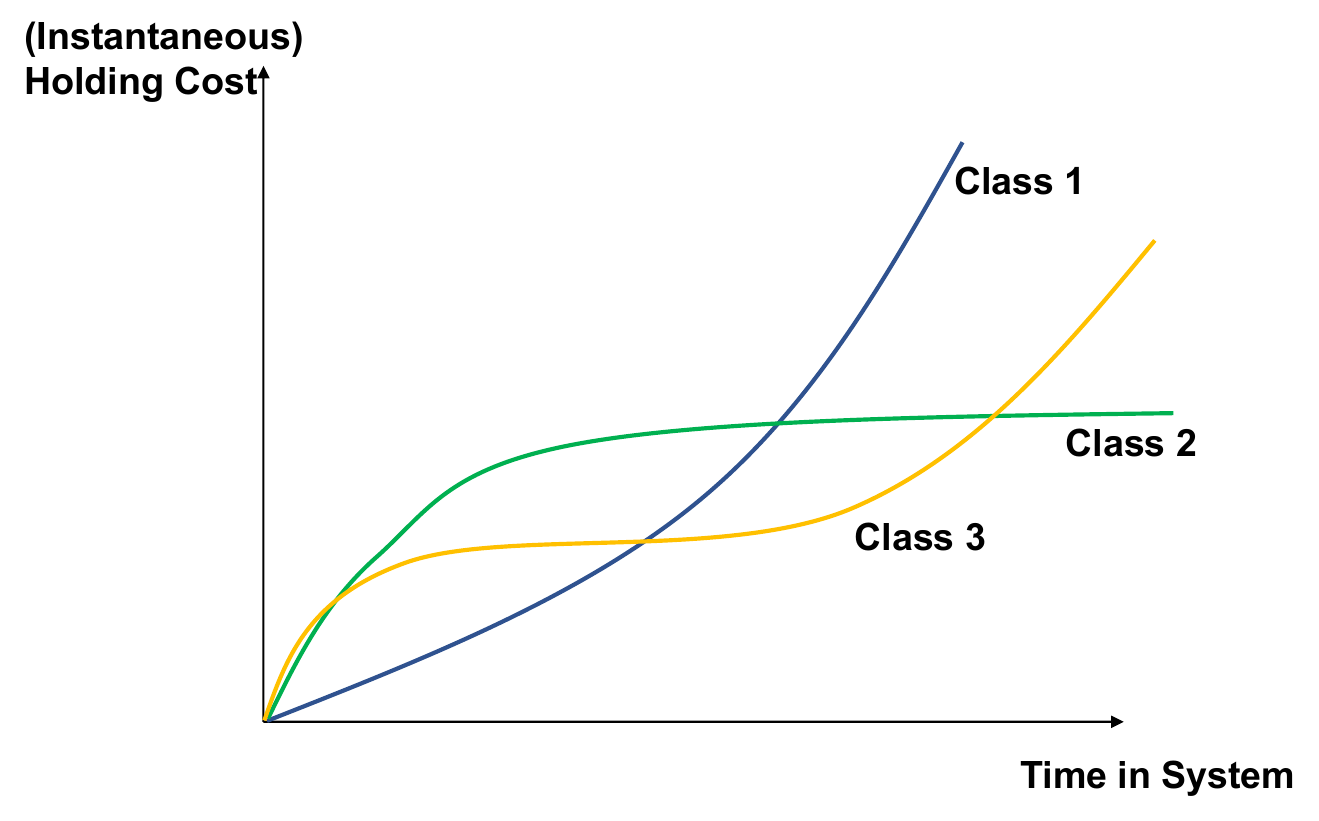}
        \caption{Classes with different holding costs.}
        \label{fig:different hold}
    \end{minipage}\hfill
    \begin{minipage}[c]{0.5\linewidth}
        \centering
        \includegraphics[width=0.8\textwidth]{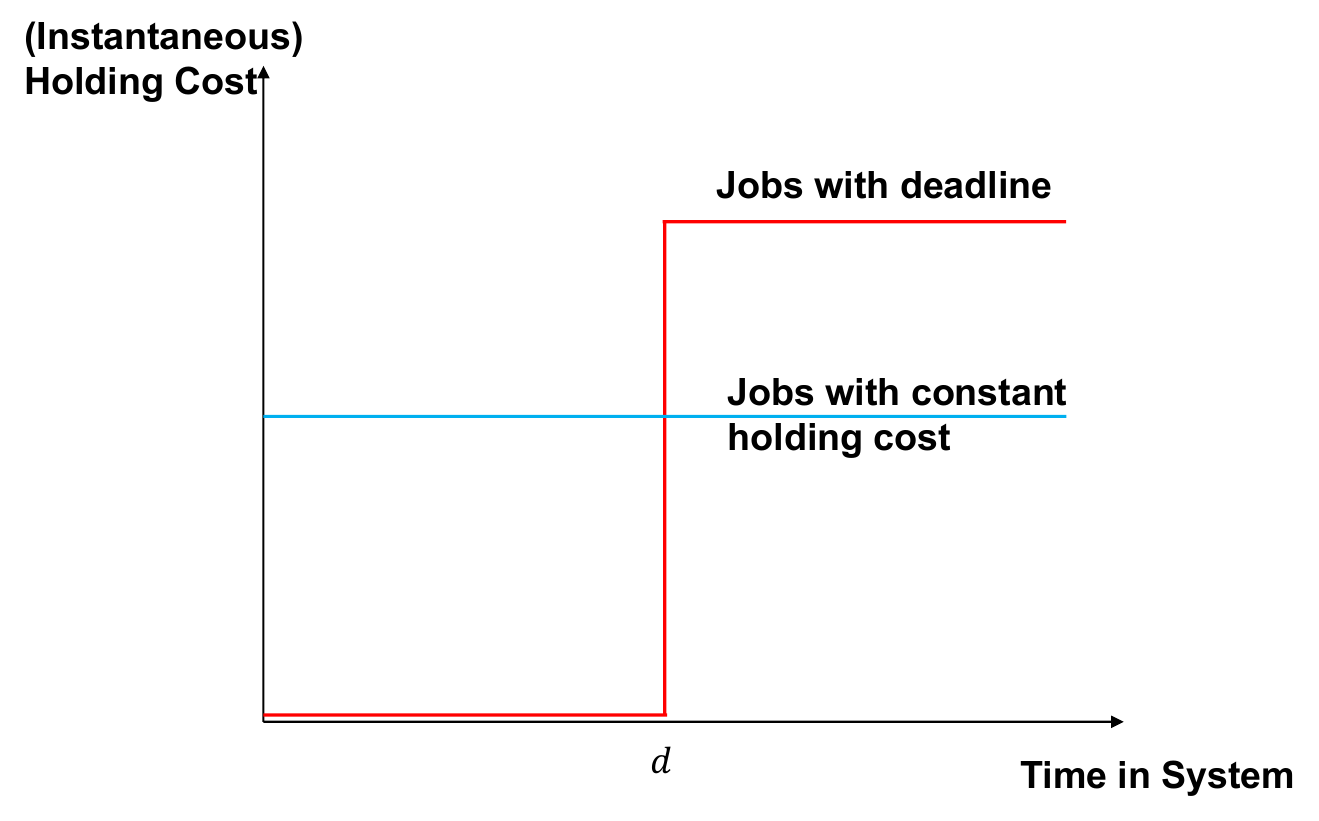}
        \caption{Example: generalized $c\mu$ rule is suboptimal.}
        \label{fig: example1}
    \end{minipage}\hfill
    \end{figure}

\subsection{Prior attempts to improve upon the generalized $c \mu$ rule}

 One of the difficulties in solving the TVHC problem lies in the fact that a job's holding cost depends on its age.  Hence, we need to track the ages of all jobs in the system, which requires an {\em infinite-dimensional} state space. As a result, the existing literature has considered  more tractable versions of the TVHC problem that have a {\em finite-dimensional} state space. Such a finite-dimensional state space allows the authors to represent their problem as a finite-dimensional Markov Decision Problem (MDP), or, more specifically, a Restless Multi-Armed Bandit problem (R-MAB) with a finite number of arms. We provide a brief tutorial on R-MABs in Section~\ref{sec:2.1}.

One way to create a finite-dimensional state space for the TVHC problem is to assume a static setting, where all $n$ jobs are present at time $0$ and there are no new arrivals.  This is the approach taken in~\cite{anand2018whittle,aalto2024whittle}. By limiting the number of jobs, the problem can now be translated to an $n$-arm R-MAB. From this R-MAB the authors then derive a Whittle index, which determines which job to run at every moment in time (see Appendix~\ref{sec:2.2} for a tutorial on the Whittle index). Figure~\ref{fig:road1} shows a road-map of the solution. The drawback of the static version setting is that the arrival rate (and consequently the load) of each class cannot be incorporated in the policy.  This is unfortunate, because, for example, in Figure~\ref{fig: example1}, it is reasonable to expect that if load is higher, we might want to start working on the deadline-oriented jobs sooner than we would under lower load.  

\begin{figure}[h]
    \centering 
    \scalebox{1}{ 
    \includegraphics[width=0.75\textwidth]{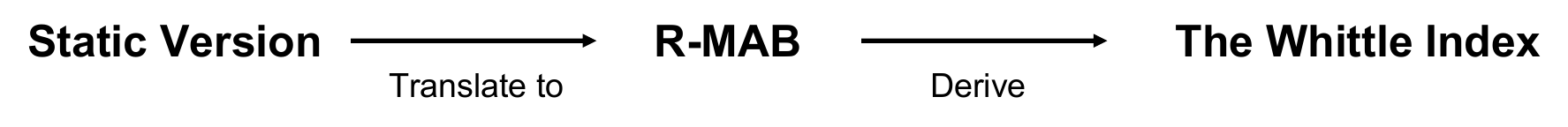}}
    \caption{A road-map to derive the Whittle Index policy for the static version of our problem: The static version is first translated to a R-MAB problem. Then the Whittle Index is derived based on the R-MAB problem.}
    \label{fig:road1}
\end{figure}


Another way to create a finite-dimensional state space is to change the TVHC problem so that individual jobs no longer have a holding cost.  Instead, the total holding cost for class $i$ is a function of the {\em number} of class $i$ jobs present, see  \cite{gurvich2009scheduling, atar2004scheduling,bispo2013single,larrnaaga2016dynamic,ansell2003whittle, larranaga2014index, glazebrook2003index}. With this change, one only needs to track the number of jobs within each class, enabling the problem to be translated into a finite-dimensional state space MDP or a $k$-arm R-MAB. 
Several papers, \cite{ansell2003whittle, larranaga2014index, glazebrook2003index}, next follow the road-map in Figure~\ref{fig:road2} where they use the $k$-arm R-MAB to derive a Whittle index policy. This queue-length holding cost setting is complementary to our age-based holding cost setting, as the policy derived for one setting cannot be translated into a policy in the other setting. 

\begin{figure}[h]
    \centering 
    \scalebox{1}{ 
    \includegraphics[width=0.75\textwidth]{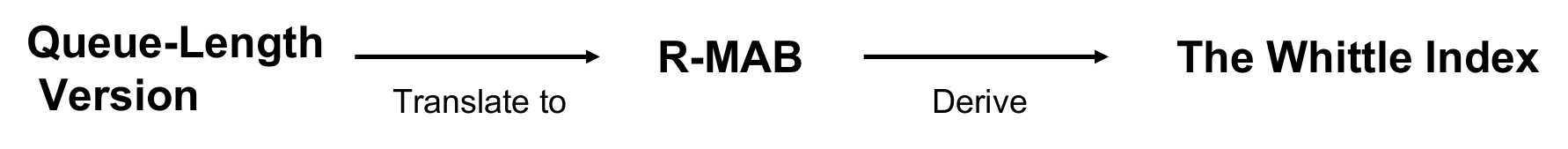}}
    \caption{A road-map 
    used to derive the Whittle Index policy for the queue-length holding cost setting.}
    \label{fig:road2}
\end{figure}

\vspace{-.1in}
\subsection{Our approach and contributions} 
The Whittle Index policy is known to be a good heuristic for R-MAB problems (see \cite{nino2007dynamic} for a discussion).
Like the prior work of Figures~\ref{fig:road1} and~\ref{fig:road2}, we adopt a similar road-map to derive the Whittle Index for our age-based holding cost minimization problem.   Unfortunately, in our problem we have infinite-dimensional state space which makes things much harder.\footnote{In fact, both papers on the static version~\cite{anand2018whittle,aalto2024whittle} note the analytical {\em intractability} of the dynamic age-based TVHC setting.}
First, while the translation to a tractable R-MAB is straightforward in prior works, it is hard to reduce our problem to a finite-arm R-MAB. Second, as in all the prior works, even after reducing to an R-MAB problem, we still need to derive the Whittle Index from the R-MAB, which requires overcoming two obstacles: 
\textit{(i)}~establishing \emph{indexability}, a key property of the R-MAB ensuring that the Whittle Index is well-defined; and 
\textit{(ii)}~performing the actual derivation of the Whittle index, which can be intricate.

In overcoming these difficulties, our paper makes two primary contributions, outlined in Figure~\ref{fig:road3}: First, 
Theorem~\ref{thm:bandit} (see Section~\ref{sec:restless}) 
translates the TVHC problem to a novel R-MAB problem. Second, Theorem~\ref{thm:index} (see Section~\ref{sec:Whittle}) proves indexability and derives the Whittle Index for our R-MAB problem, thus yielding a Whittle-based heuristic policy for our problem. 

\begin{figure}[h]
    \centering 
    \scalebox{1}{ 
    \includegraphics[width=0.75\textwidth]{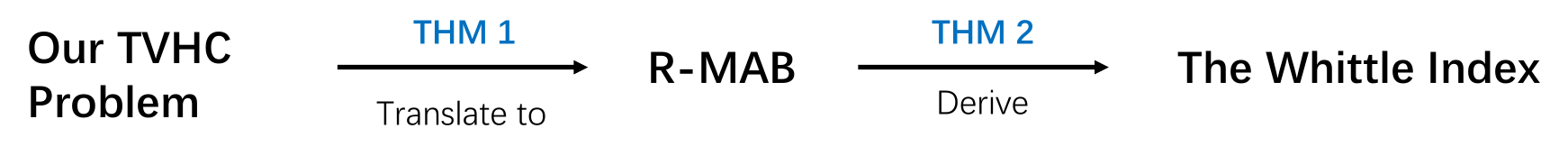}}
    \caption{Using Theorem~\ref{thm:bandit} and Theorem~\ref{thm:index}, we follow this road-map to derive the Whittle Index policy.  }
   \label{fig:road3}
\end{figure}

Our final result is a Whittle-based policy which 
always preemptively runs the job with the highest Whittle Index. 
The Whittle Index has an elegant and intuitive formulation given in Corollary~\ref{cor:main} (see Section~\ref{sec:Whittle}) and repeated here: Our Whittle Index for a class $i$ job with age $t$, $W_i(t)$, is given by 
\begin{equation}
    W_i(t) = \mu_i \cdot \E{c_i(t+X)}, \quad \text{where } X\sim \Exp(\mu_i-\lambda_i).
    \label{eq: intro Whittle}
\end{equation}

Intuitively, while the generalized $c\mu$ rule focuses on the current holding cost $c_i(t)$ and schedules according to the index $\mu_i c_i(t)$, our Whittle index looks a little further into the future to age $t + X$. Note that if the load is high ($\lambda_i$ is close to $\mu_i$), then we look further into the future. This aligns with the motivating example in Figure~\ref{fig: example1}, where it is preferable to prioritize a job before its deadline, particularly when there are likely to be many jobs in the system.

It is also interesting to contrast our policy with the heuristic given in \cite{aalto2024whittle}, which is the Whittle index in the static version. Under exponential job sizes, their index has the form $\mu_i \cdot \E{c_i(t+Y)}$, where $Y\sim \Exp(\mu_i)$. Note that our index given in \eqref{eq: intro Whittle} matches theirs when $\lambda_i=0$, where our setting degenerates to the static setting. In Section~\ref{sec:simulation} we compare our policy with all the existing heuristics (including  \cite{aalto2024whittle}). Our simulations show that our policy outperforms all the other heuristics.

\section{Tutorial on (Markov) R-MAB Problem}
\label{sec:2.1}
In this paper, when we say MAB, we always refer to the {\em Markov} MAB: 
In a Markov MAB problem, each arm corresponds to a Markov process with states that evolve based on an underlying state transition model. The agent chooses to pull an arm at each time step, incurring a cost determined by the states of the arms. The objective is to minimize the cumulative cost over time. 

Mathematically, let \( k \) denote the number of arms. Each arm \( i \) is associated with a Markov Decision Process (MDP), which is defined by a state space, the action space $\{active, passive\}$, the transition probabilities given states and actions, and a cost function $c_i(s)$ representing the cost incurred when arm $i$ is in state $s$. At each time step $t$, the agent observes the current states of all arms and selects at most one arm to pull (pulling means choosing the action $active$). Cost is incurred for each arm, and the state of each arm evolves according to its respective MDP and actions. 
The objective is to minimize the long-run average cost over an infinite time horizon.

A {\em discounted} MAB is defined exactly the same except for the objective. Instead of the long-run average cost, a discounted MAB aims to minimize the cumulative discounted cost, where the cost at time $t$ is multiplied by $e^{-\alpha t}$, and $\alpha>0$ is called the discount factor. 

A MAB is called {\em restful} if the state of an arm does not evolve when the arm is not pulled (the action is $passive$).  In contrast, An MAB is called {\em restless} if the state evolves whether the arm is pulled or not. For a restful MAB, the optimal policy is the Gittins index policy (\cite{gittins1979bandit,scully2021gittinspolicymg1queue}), while the optimal policy for restless MAB (R-MAB) is open.
The TVHC problem is intrinsically restless. 


\section{Translation to a Discounted R-MAB}
\label{sec:restless}

The main goal of this section is to complete the first step in the road-map (Figure~\ref{fig:road3}), which is the translation from the TVHC problem to a discounted R-MAB problem. 
First in Section~\ref{sec:scheduling formulation}, we restate the TVHC problem setting, and prove that one should serve jobs within each class in FCFS order. This fact allows us to only consider index functions that enforce FCFS within each class, which we will see is essential for the translation to an R-MAB problem. 
Second in Section~\ref{sec:restless bandit}, we prove our theorem which translates the TVHC problem to an R-MAB problem. Finally in Section~\ref{sec:decay}, a discount factor is introduced to the R-MAB problem. The discounted R-MAB problem serves as the starting point of the typical Whittle Index approach. The Whittle Index for the R-MAB without discounting is obtained by taking the limit on the discount factor (see Section~\ref{sec:decay}).


\subsection{The TVHC problem}
\label{sec:scheduling formulation}

\subsubsection{Problem setting}
We briefly restate the TVHC problem setting as follows:
In a single-server system, there are $k$ types of jobs. For each type $i$, we assume a Poisson arrival process with rate $\lambda_i$. We also assume that the job size distribution is exponentially distributed with rate $\mu_i$.
For a type $i$ job of age $t$, its instantaneous holding cost is $c_i(t)$, which is a non-decreasing and smooth\footnote{We need smoothness to simplify our derivation. However, our policy applies to any holding cost function that can be arbitrarily closely approximated by a smooth function (e.g., a deadline function, or any continuous function). } function. The objective of the scheduler is to schedule the jobs in order to minimize the time-average mean holding cost. Specifically,  Let $c_{total}(s)$ denote the sum of the holding costs of all jobs currently in the system at time $s$. Then the objective is to minimize $\E{\mbox{Holding Cost}}$, where
$$\E{\mbox{Holding Cost}} = \lim_{t \to \infty} \frac{1}{t} \int_{s = 0}^t c_{total}(s) ds. $$

To make the problem well-defined, 
we assume that there exists a policy to make the time-average mean holding cost converge. 



    


\subsubsection{Index functions}
We only focus on index policies in this paper. These are policies where a job's index (priority level) only depends on its own state, not on the state of other jobs in the system.  Specifically, since the job sizes are exponential, a job's index only depends on the job's type and age and does not depend on its attained service. Thus an index policy is specified by a set of functions $\{V_i(\cdot)\}_{i=1}^k$, where $V_i(t)$ is the index of a job of type $i$ with age $t$.

\subsubsection{FCFS within each class}
Intuitively, since the holding cost is increasing and the job sizes are exponentially distributed, to minimize mean holding cost, we should schedule jobs within each type in FCFS (First Come First Serve) order , since earlier arriving jobs have higher cost. Mathematically, we have the following lemma:

\begin{lemma}[FCFS within each type is optimal]
\label{lemma:FCFS}
    The optimal policy must serve jobs within each type in FCFS order. 
\end{lemma}

\begin{proof}
See Appendix~\ref{app:lemma3.1}.
\end{proof}

Motivated by Lemma~\ref{lemma:FCFS}, throughout this paper, we focus on index policies that enforce FCFS order within each class. This is equivalent to having a non-decreasing index function for each class with a FCFS tie-breaking rule.
In Section~\ref{sec:restless bandit}, we will leverage the fact that the index functions enforce FCFS within each class to substantially reduce the dimensionality of the state space.

\subsection{Translation to an R-MAB problem}
\label{sec:restless bandit}

In this section, we translate the TVHC problem into an R-MAB problem. Intuitively, our key idea is to let each arm of the R-MAB problem track the age of the oldest job within each class in the TVHC problem.  At first this seems insufficent, because we're not capturing the state of all the other jobs within each class.  
However we will show how the FCFS ordering within each class ensures that the distribution of the ages of younger jobs can be effectively captured through the stochastic behavior of the Poisson arrival process (a generally similar, yet distinct, idea is applied to a different problem in \cite{stanford2014waiting}). Our key theorem is as follows:
\begin{customthm}[Translation to an R-MAB]
\label{thm:bandit}
There exists an R-MAB problem, such that
for any set of non-decreasing index functions $\{V_i(\cdot)\}_{i=1}^k$, the corresponding index policies (breaking ties by FCFS) in the TVHC problem and the R-MAB problem incur the same cost.
\end{customthm}
\begin{proof}[Proof Sketch]

We construct a continuous R-MAB as follows. 
\begin{itemize}
    \item The R-MAB has $k$ arms, each representing a class. There are two actions for each arm (active or passive), where arm $i$ is active means the oldest class $i$ job is served, and passive means the job served at this moment is of some other type.   
    \item {\em Arm state} $T_i(t)$: The $i$th arm's state is $T_i(t)\in \mathbb{R}$. The arm state can be interpreted in the TVHC problem as the age of the oldest type $i$ job at time $t$.  If there is no type $i$ job in the system, $T_i(t)$ is negative, which means the next type $i$ arrival happens $-T_i(t)$ time later. Note that for arm $i$, the action active is only allowed at time $t$ when $T_i(t)$ is positive. 
    \item {\em Transition Probability: } If the action for arm $i$ is passive, the $i$th arm state grows with rate 1. Otherwise if arm $i$ is active, the $i$th arm state may drop an $\Exp(\lambda_i)$ amount according to a Poisson process (when completions happen). Mathematically, the transition function is 
    \[\text{passive: }dT_i(t) = dt, \quad \text{active: }dT_i(t) = dt - X \cdot dN_{\mu_i}(t),  \]
    $\text{where } X\sim \Exp(\lambda_i), \ \text{and} \ N_{\mu_i}  \text{ is a Poisson counting process with rate $\mu_i$.}$

    This transition function can be interpreted in the TVHC problem as follows: A passive action means that the oldest class $i$ job is not in service and its age grows with rate 1. If the action is active, then in the next $dt$ time period, the oldest class $i$ job is served. There is a probability of $\mu_i dt$ that the job is completed and leaves the system, in which case the oldest class $i$ job in the system becomes the previously second-oldest class $i$ job. Since the inter-arrival time follows the distribution $\Exp(\lambda_i)$, the age of the oldest job drops by an $\Exp(\lambda_i)$ amount.

    \item {\em Constraint: } The number of active arms at any time is at most 1.

    \item {\em Cost Function $r_i(s)$: } Arm $i$ incurs a cost of $r_i(s)$ at state $s$, where $r_i(s)$ is defined to be 
    \begin{equation}
        r_i(s):= c_i(s) + \E{\sum_{j=1}^{N_i} c_i(Y_j)}, 
        \label{eq:r}
    \end{equation}
    where $Y_j=\sum_{m=1}^j X_m$, $X_m \sim \Exp(\lambda_i)$, and $N_i$  is the random variable such that $Y_{N_i}<s$ and $Y_{N_i+1}\geq s$. For $s<0$, define $r_i(s)=0$.

    To interpret (\ref{eq:r}), observe that $r_i(s)$ represents the expected total instantaneous holding cost of all class $i$ jobs given that the age of the oldest class $i$ job is $s$: Since the index policy serves class $i$ jobs in FCFS order, no ``young'' class $i$ jobs (class $i$ jobs younger than the current oldest one) have been completed. Since the inter-arrival times are distributed as $\Exp(\lambda_i)$, their ages are distributed as 
    \[s-X_1, \ s-X_1-X_2, \ \ldots, \ s-\sum_{m=1}^{N_i} X_m.\]

    Thus the expected total instantaneous cost of all the young class $i$ jobs is $\E{\sum_{j=1}^{N_i} c_i(s-Y_j)}$, which is equal to $\E{\sum_{j=1}^{N_i} c_i(Y_j)}$ as the $N_i$ arrivals are distributed as uniform order statistics in $(0,s)$.
    
    

    \item {\em Objective: } The objective is to minimize the long-run expected cost. Mathematically, 
    \[\cost = \E{\lim_{x\to \infty} \frac{1}{x}\int_0^x \sum_{i=1}^k r_i(T_i(t)) dt \  \bigg |\ T_i(0) = 0}.\]
\end{itemize}

Through the construction and the corresponding interpretation in the TVHC problem, it is straightforward to see a coupling between the $i$th arm state and the age of the oldest class $i$ job 
in the TVHC problem. However, although intuitive, it is delicate to rigorously prove that characterizing the holding cost of all young jobs by expectation does not change the expected long-run mean cost. We refer to  Appendix~\ref{app:thm3.2} for the rigorous proof, where we use a coupling argument over sample paths to build equivalence between the two problems. 
\end{proof}

\subsection{Introducing the discount factor}
\label{sec:decay}
Finally, we introduce a discount factor into our R-MAB problem. It is a typical step before applying the Whittle Index approach. The discounted objective is defined to be:
\begin{equation}
\cost(\mathbf{s_0}, \alpha) = \E{\int_0^\infty \sum_{i=1}^k r_i(T_i(t))\cdot \alpha e^{-\alpha t} dt \  \bigg |\ T_i(0) = \mathbf{s_0}(i)},
\label{eq: decay cost}
\end{equation}
where $\mathbf{s_0}$ is the vector of the initial states and $\mathbf{s_0}(i)$ is the initial state of arm $i$.

By standard results in dynamic programming (e.g., \cite{ansell2003whittle, puterman2014markov}), also known as the Wiener's Tauberian theorem, we have the following lemma:

\begin{lemma}
\label{lemma:decay}
    For any initial state $\mathbf{s_0}$, 
    \[\lim_{\alpha\to 0} \cost(\mathbf{s_0},\alpha) = \cost.\]
\end{lemma}

This lemma indicates that to get a good index for our R-MAB problem (and thus the TVHC problem), it suffices to find a good index policy to optimize the discounted cost $\cost(\mathbf{s_0},\alpha)$, and then take the limit $\alpha\to 0$ on the index.

\section{Derivation of the Whittle Index}
\label{sec:Whittle}

Given the discounted R-MAB problem, we now perform the second part of the road-map of Figure~\ref{fig:road3}: the derivation of the Whittle Index. Our main theorem is Theorem~\ref{thm:index}.  In Corollary~\ref{cor:main} we obtain the Whittle Index for our TVHC problem by taking the limit on the discount factor in Theorem~\ref{thm:index}.

The section is organized as follows: First in Section~\ref{sec:single}, we define the single-arm bandit formulation and define the Whittle Index. Then in Section~\ref{sec:def of whittle} we propose a guess of the Whittle Index. The closed form of the guessed Whittle Index is derived in Section~\ref{sec:solve whittle}. Finally, we verify indexability and that this guess of the Whittle Index is true in Section~\ref{sec:index}, hence proving Theorem~\ref{thm:index}. 

In most parts of this section, we focus on the single-arm bandit problem, and we drop the subscript $i$ for simplicity. 
However, we should remember that there is an independent single-arm bandit problem for each arm in our R-MAB problem, and the Whittle Index obtained from each single-arm bandit problem serves as the index for each arm in the R-MAB problem.


\subsection{Set up the Whittle Index approach}
\label{sec:single}

In this section, we set up the problem following the Whittle approach (see Appendix~\ref{sec:2.2} for a tutorial). 
We first define the single-arm bandit for class $i$ where not pulling the arm is rewarded with some compensation. Based on this single-arm bandit formulation, the Whittle Index is defined. 

Formally, we drop the subscript $i$ and redefine our problem as follows.

\begin{itemize}
    \item {\em State and Action: } As before, the state is denoted by $T(t)$. The action set is $\{\text{active},\text{ passive}\}$, and ``active'' is only allowed at time $t$ if $T(t)>0$. 
    \item {\em Transition Probability: }
    \begin{equation}
    \label{eq:active trans}
        \text{passive: }dT(t) = dt,\quad \text{active: }dT(t) = dt - X \cdot dN_\mu(t),
    \end{equation}
    $\text{ where } X\sim \Exp(\lambda) $, and $ N_{\mu} \text{ is a Poisson counting process with rate $\mu$.}$

    \item {\em Passive Compensation $\ell$: } We define $\ell$ to be the compensation for a passive action. It is further explained in the cost function below.

    \item {\em Cost Function $\theta(s,a): $} We define
     the cost incurred at state $s$ and action $a$ to be
     \[
\theta(s,a) := 
\begin{cases}
    \alpha\cdot r(s), & \text{if } a=\text{active}, \\
    \alpha\cdot r(s)-\ell, & \text{if } a=\text{passive}, 
\end{cases}
\]
     where $r(s)$ is defined as before:
    \begin{equation}
        r(s):= c(s) + \E{\sum_{j=1}^N c(Y_j)}, 
        \label{eq:r2}
    \end{equation}
    where $Y_j=\sum_{m=1}^j X_m$, $X_m \sim \Exp(\lambda)$, and $N$ is the random variable such that $Y_N<s$ and $Y_{N+1}\geq s$. For $s<0$, define $r(s)=0$.

    \item {\em Objective: } The objective is to minimize the discounted total holding cost. Mathematically, let $s_0$ denote the initial state. Then the discounted total holding cost is 
    \[\cost(s_0, \alpha) = \E{\int_0^\infty \theta(T(t))\cdot e^{-\alpha t} dt \  \bigg |\ T(0) = s_0}.\]
\end{itemize}

A policy $\pi$ takes a state $T$ and outputs an action $a$, where $a$ is active or passive.  Thus we denote a policy by a function $\pi:(-\infty,\infty)\to \{0,1\}$ where 1 means active. Note that any policy receiving a negative state must return the passive action.

For any value of the compensation $\ell$, consider the optimal policy of the single-arm bandit problem. Define $\Pi(\ell)$ to be the set of states where the passive action is optimal. Note that for any compensation $\ell$, any non-positive state is in $\Pi(\ell)$.

\begin{definition}[$\Pi$]
    For any compensation $\ell$, define \[\Pi(\ell):=\{t_0 \mid \text{passive action at state $t_0$ is optimal}\}.\]
    Throughout, we omit writing the discount factor $\alpha$, to simplify notation.
\end{definition}

Intuitively, the larger the compensation, the more likely it is that a passive action should be taken. This property is called ``indexability.''

\begin{definition}[Indexability]
 A problem is indexable if for any $\ell_1<\ell_2$, we have $\Pi(\ell_1)\subseteq \Pi(\ell_2)$.
\end{definition}

The Whittle Index of a state is defined to be the smallest compensation such that the passive action is optimal at this state. It can be intuitively understood as the value of compensation where active and passive actions are both optimal.
\begin{definition}[Whittle's Index]
The Whittle Index of a state $t$ with discount factor $\alpha$ is defined to be
\[W(t,\alpha):= \inf_\ell \{t\in \Pi(\ell)\}.\]
\end{definition}





Notice that both indexability and the form of the Whittle Index involves understanding $\Pi(\ell)$, which requires us to understand the optimal policy of this single-arm bandit formulation. We follow these steps to establish indexability and derive the Whittle Index: \textit{(i)} We propose a guess of the optimal policy, and accordingly define the guessed Whittle Index (Section~\ref{sec:def of whittle}); \textit{(ii)} We compute the guessed Whittle Index (Section~\ref{sec:solve whittle}); \textit{(iii)} We verify indexability and that the guessed Whittle Index is the true Whittle Index (Section~\ref{sec:index}).

\subsection{A guess of the Whittle Index}
\label{sec:def of whittle}

Intuitively, it is reasonable to guess that the optimal policy is a threshold policy that chooses the passive action for states smaller than the threshold and chooses active otherwise. Mathematically, we define a threshold policy $Threshold(x)$ below:
\begin{definition}[$Threshold(x)$]
    A threshold policy, $Threshold(x)$, parameterized by state $x$, selects the passive action iff state $t<x$. 
\end{definition}

For any threshold policy, denote its cost by the following notation:
\begin{definition}[Cost of $Threshold(x)$]
\label{def:costxt}
    For policy $Threshold(x)$, define the discounted total holding cost given discount factor $\alpha$, initial state $s_0$ and compensation $\ell$ to be $\cost^x(s_0,\alpha,\ell)$.
\end{definition}

With the guess that a threshold policy is optimal, the following ``guessed Whittle Index'' is proposed. Both guesses will be proven true in Section~\ref{sec:index}.

\begin{definition}[Guessed Whittle's Index]
    The guessed Whittle Index of state $t_0$, $\guess(t_0,\alpha)$, is the value of the compensation $\ell$ such that the policy $Threshold(t_0)$ is the optimal threshold policy for initial state $t_0$. Mathematically, define $\guess(0,\alpha) = 0$ and for any $t_0>0$, define $\guess(t_0,\alpha)$ to be the value of $\ell$ that satisfies \footnote{Intuitively, the optimal Threshold policy is the value of $x$ that minimizes $Cost^x(t_0,\alpha,\ell)$. Differentiating yields \eqref{eq:whittle}.}
    \begin{equation}
    \lim_{\delta\to 0} \frac{\cost^{t_0}(t_0,\alpha,\ell) - \cost^{t_0+\delta}(t_0,\alpha,\ell)}{\delta} = 0.
    \label{eq:whittle}
\end{equation}
\end{definition}

Unfortunately, it is not easy to solve the guessed Whittle Index from \eqref{eq:whittle}. The whole next subsection (Section~\ref{sec:solve whittle}) is devoted to solving \eqref{eq:whittle}.

\subsection{The Guessed Whittle Index}
\label{sec:solve whittle}

In this section we will solve (\ref{eq:whittle}) for the guessed Whittle Index, resulting in  
Lemma~\ref{thm:main}.
First we introduce some notation. Then we analyze the cost of a threshold policy (Lemma~\ref{lemma:costs}) and give a (complicated) expression for the guessed Whittle Index (Lemma~\ref{lemma:solving whittle}). Finally, we show how to simplify the expression to the elegant and simple formula in Lemma~\ref{thm:main}.

\subsubsection{Notation}

A short summary of the notation used is given in Table~\ref{table:notation}. 

{\small 
\begin{table}[h]
  \caption{Additional Notation Table}
  \label{tab:freq}
  \begin{tabular}{ccl}
    \toprule
    Notation  & Meaning\\
    \midrule
    $\costbar(t_0,\alpha)$ &  The discounted cost incurred during an M/M/1 busy period\\
    $T_1$ & $T_1\sim BP^{M/M/1}$, the length of an M/M/1 busy period \\
    $T_2$ & $T_2\sim \Exp(\lambda)$ \\
    $\gamma_i,\ i=1,2$ & $\E{e^{-\alpha T_i}}$\\
    
    $\Gamma(t_0,\alpha)$ & $\E{\int_0^{T_2}\alpha r(t_0-T_2+t) e^{-\alpha t} dt}$\\
    $X$ & $X\sim \Exp\left(\frac{\alpha}{1-\gamma_1}\right)$\\
  \bottomrule
\end{tabular}
\label{table:notation}
\end{table}
}

\begin{definition}[$\costbar$]
\label{def:decay costbar}
    In an M/M/1 queue, suppose a busy period starts at time 0, and the age of the oldest job in the system is defined to be $A(t)$. Define $T_1$ to be the first time such that $A(t)<0$ (which means there is no job in the system and the busy period ends). 
    Define $\costbar(t_0,\alpha)$ to be the expected discounted cost incurred during the busy period where the cost at time $t$ is $r(t_0+A(t))$:
    \begin{equation}
        \costbar(t_0,\alpha):= \E{\int_0^{T_1} r(t_0+A(t))\cdot \alpha e^{-\alpha t} dt}.
    \end{equation}
\end{definition}

\begin{definition}[$\Gamma(t_0,\alpha)$]
\label{def:decay gamma}
Let $T_2\sim \Exp(\lambda)$. We define $\Gamma(t_0,\alpha)$ to be the expected discounted cost during the interval $[t_0 - T_2, t_0]$: 
    \begin{equation}
        \Gamma(t_0,\alpha):= \E{\int_0^{T_2} r(t_0-T_2+t) \cdot \alpha e^{-\alpha t} dt}.
    \end{equation}
\end{definition}

\begin{definition}[$\gamma_1,\gamma_2$]
    Define $\gamma_1:=\E{e^{-\alpha T_1}}$ (where $T_1\sim  BP^{M/M/1}$ is the length of an M/M/1 busy period) and $\gamma_2:=\E{e^{-\alpha T_2}}$ (where $T_2\sim \Exp(\lambda)$).
\end{definition}

We have the following lemma characterizing $\gamma_1$ and $\gamma_2$.

\begin{lemma}
    $\gamma_1$ and $\gamma_2$ satisfy the following equations:
    \begin{equation}
    \mu = \frac{(\alpha+\lambda-\lambda\gamma_1)\gamma_1}{1-\gamma_1},
    \qquad \qquad 
    \gamma_2=\frac{\lambda}{\lambda+\alpha}.
    \label{eq:gamma2}
\end{equation}
\end{lemma}
\begin{proof}
    Since $\gamma_1 =\E{e^{-\alpha T_1}} = \widetilde{T_1}(\alpha)$, where $T_1$ is the length of an M/M/1 busy period, we have the following equation from classic queueing theory (e.g., see Section 27.2 in \cite{harchol2013performance}):
\[\gamma_1 = \frac{\mu}{\mu+\alpha+\lambda-\lambda \gamma_1},\]
yielding the first equation.  Likewise, the second equation comes from: 
\[\gamma_2=\E{e^{-\alpha T_2}} = \widetilde{T_2}(\alpha)=\frac{\lambda}{\lambda+\alpha}.\]

\end{proof}

\subsubsection{An expression for the guessed Whittle Index}
Now we work on solving (\ref{eq:whittle}) for the guessed Whittle Index. In this subsection, through analyzing the cost of a threshold policy (Lemma~\ref{lemma:costs}), a complicated expression for the guessed Whittle Index is derived (Lemma~\ref{lemma:solving whittle}). 

We first analyze the cost incurred by the $Threshold(t_0 + \delta)$ policy. Recall that the cost of a threshold policy is defined in Definition~\ref{def:costxt}.
\begin{lemma}
    \label{lemma:costs}
    For any $\delta\geq 0$, we have that
    \begin{align}
        \cost^{t_0+\delta}(t_0,\alpha,\ell) =& \frac{1}{1-\gamma_1\gamma_2}\bigg(\delta \cdot (\alpha\cdot r(t_0)-\ell) + e^{-\alpha\delta}\cdot \costbar(t_0+\delta,\alpha) \nonumber \\
        &+ e^{-\alpha\delta}\gamma_1\cdot \big(\Gamma(t_0+\delta,\alpha) - \delta\alpha\gamma_2 r(t_0) - \ell\cdot \frac{1}{\alpha}(1-\gamma_2\cdot e^{\alpha \delta})\big) \bigg)+ o(\delta).\label{eq:cost expr}
    \end{align}
\end{lemma}

\begin{proof}
    We consider the state transition process of policy $Threshold(t_0+\delta)$ from the initial state $t_0$ until the next time the state returns to $t_0$. As shown in Figure~\ref{fig: cost illus}, there are three phases:
    \begin{description}
        \item[Phase 1] (the $[0,\delta]$ time period): The policy stays passive, and the cost incurred at every moment is $\alpha r(t)-\ell$. Thus the incurred discounted cost during phase 1 is
        \begin{equation}
            \delta \cdot (\alpha\cdot r(t_0)-\ell)+o(\delta).
            \label{eq:costphase1}
        \end{equation}
        
        \item[Phase 2] (the time from $\delta$ until the next time that the state is below $t_0+\delta$): 
        Define $T_1$ to be the length of phase 2. Since the policy stays active from time $\delta$, we know that the dynamics of the state is exactly the same as that of the oldest age in an M/M/1 queue. Thus the time period from $\delta$ until there is no job with age more than $t_0+\delta$ is the length of a busy period in an M/M/1, i.e., 
        \begin{equation}
            \label{eq:T1}
            T_1\sim  BP^{M/M/1}.
        \end{equation}
        Moreover, we know that the cost incurred during this time period is $\costbar(t_0+\delta,\alpha)$. Note that there is a discount factor of $e^{-\alpha \delta}$ after phase 1. Thus the expected total discounted cost incurred during phase 2 is 
        \begin{equation}
            \label{eq:costphase2}
            e^{-\alpha\delta}\cdot \costbar(t_0+\delta,\alpha).
        \end{equation}

        \item[Phase 3] (time until the state returns to $t_0$): 
        At time $T_1+\delta$, the state drops from being above $t_0+\delta$ to below. 
        Since the drop is exponential with rate $\lambda$ (see (\ref{eq:active trans})), we know that the state at time $T_1+\delta$ is lower than $t_0+\delta$ by an exponential overshoot, which is $\Exp(\lambda)$. We define the state at $T_1+\delta$ to be $t_0+\delta - T_2$, where $T_2\sim \Exp(\lambda)$. 

        Now since the policy $Threshold(t_0+\delta)$ stays passive when the state is smaller than $t_0+\delta$, phase 3 lasts for $T_2-\delta$ time before the state returns to $t_0$ (i.e., the age grows to $t_0$). Thus the total discounted cost incurred during phase 3 is 
        \begin{equation}
            \label{eq:costphase3}
            e^{-\alpha(\delta+T_1)}\cdot \int_{0}^{T_2-\delta} (\alpha r(t_0+\delta - T_2+t) - \ell) e^{-\alpha t}  dt.
        \end{equation}
        
    \end{description}

    \begin{figure}[h]
        \centering
        \scalebox{1}{ 
        \includegraphics[width=0.7\textwidth]{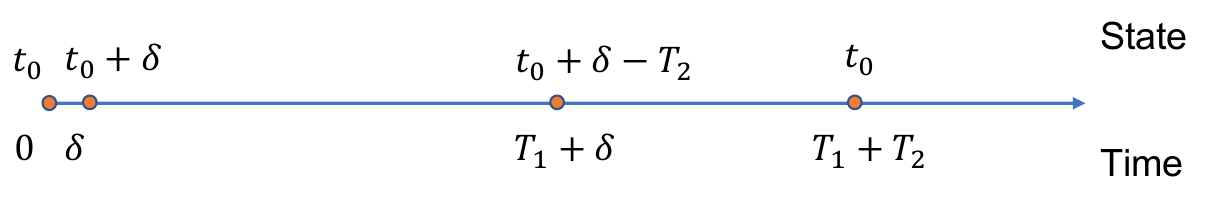}}
        \caption{Illustration state transition process.}
        \label{fig: cost illus}
    \end{figure}

    After three phases, the state returns to $t_0$. Thus summing up (\ref{eq:costphase1}), (\ref{eq:costphase2}), (\ref{eq:costphase3}) and taking the expectation gives the following equation:
    \begin{align}
        \cost^{t_0+\delta}(t_0,\alpha,\ell) = &\  \delta \cdot (\alpha\cdot r(t_0)-\ell) + o(\delta)\\
        &+ e^{-\alpha\delta}\cdot \costbar(t_0+\delta,\alpha) \\
        &+ e^{-\alpha\delta}\E{e^{-\alpha T_1}}\cdot \E{\int_{0}^{T_2-\delta} (\alpha r(t_0+\delta - T_2+t) - \ell) e^{\alpha t}  dt} 
        \label{eq:11}\\
        &+ \E{e^{-\alpha T_1}}\E{e^{-\alpha T_2}}\cost^{t_0+\delta}(t_0,\alpha,\ell).
    \end{align}

    We can use the notation from Table~\ref{table:notation} to simplify the terms.
    We can simplify line (\ref{eq:11}) as follows:
    \begin{align*}
        &\E{\int_{0}^{T_2-\delta} (\alpha r(t_0+\delta - T_2+t) - \ell) e^{\alpha t}  dt}\\
        &= \E{\int_0^{T_2}(\alpha r(t_0+\delta-T_2+t))e^{-\alpha t}} - \E{\int_{T_2-\delta}^{T_2}(\alpha r(t_0+\delta-T_2+t))e^{-\alpha t}} - \ell \cdot \E{\int_0^{T_2-\delta}e^{-\alpha t} dt} \\
        &= \Gamma(t_0+\delta,\alpha) - \E{\delta\alpha r(t_0) e^{-\alpha T_2} + o(\delta)} - \ell\cdot \frac{1}{\alpha}(1-\E{e^{-\alpha (T_2-\delta)}}) \\
        &= \Gamma(t_0+\delta,\alpha) - \delta\alpha r(t_0) \gamma_2 + o(\delta) - \ell\cdot \frac{1}{\alpha}(1-\gamma_2\cdot e^{\alpha \delta}).
    \end{align*}

    Replacing (\ref{eq:11}) with the above equation, and rearranging terms, gives the proof.
\end{proof}

We now use Lemma~\ref{lemma:costs} to solve equation (\ref{eq:whittle}) to get the guessed Whittle Index.

\begin{lemma}
    \label{lemma:solving whittle}
    The guessed Whittle index with discount factor $\alpha$ is given by: 
    \begin{equation} 
        \guess(t_0,\alpha)= \frac{1}{1-\gamma_1} \left(\alpha(1-\gamma_1\gamma_2) r(t_0) + \costbar'(t_0,\alpha) - \alpha \costbar(t_0,\alpha) + \gamma_1 \Gamma'(t_0,\alpha) - \alpha\gamma_1\Gamma(t_0,\alpha)\right).
        \label{eq:whittle alpha}
    \end{equation}    
\end{lemma}
\begin{proof}
    Note that (\ref{eq:whittle}) is equivalent to 
    \[\frac{d\cost^{t_0+\delta}(t_0,\alpha,\ell)}{d\delta}\bigg |_{\delta = 0}=0\]

    The left hand side is the derivative of (\ref{eq:cost expr}) with respect to $\delta$. For simplicity, we omit the $o(\delta)$ term and define
    \[\costbar'(t_0,\alpha):= \frac{d\costbar(t,\alpha)}{dt}\bigg|_{t=t_0}, \quad \text{and } \quad \Gamma'(t_0,\alpha):= \frac{d\Gamma(t,\alpha)}{dt}\bigg|_{t=t_0}.\]

Now the derivative of (\ref{eq:cost expr}) has the form: 
    
    \begin{align*}
        \frac{d\cost^{t_0+\delta}(t_0,\alpha)}{d\delta}\bigg |_{\delta = 0} =& \frac{1}{1-\gamma_1\gamma_2}\bigg( \alpha r(t_0) - \ell + \costbar'(t_0,\alpha) - \alpha \costbar(t_0,\alpha) \\
        &+ \gamma_1 \left(\Gamma'(t_0,\alpha) - \alpha\gamma_2 r(t_0) + \ell \cdot\gamma_2\right)\\
        &- \alpha\gamma_1\left( \Gamma(t_0,\alpha) - \ell\frac{1-\gamma_2}{\alpha} \right) \bigg).
    \end{align*}

    Thus, \eqref{eq:whittle} is equivalent to: 
    \[ \alpha(1-\gamma_1\gamma_2) r(t_0) + \costbar'(t_0,\alpha) - \alpha \costbar(t_0,\alpha) + \gamma_1 \Gamma'(t_0,\alpha) - \alpha\gamma_1\Gamma(t_0,\alpha) = (1 - \gamma_1)\cdot \ell,\]

    which yields:
    \[\ell=\frac{1}{1-\gamma_1} \left(\alpha(1-\gamma_1\gamma_2) r(t_0) + \costbar'(t_0,\alpha) - \alpha \costbar(t_0,\alpha) + \gamma_1 \Gamma'(t_0,\alpha) - \alpha\gamma_1\Gamma(t_0,\alpha)\right).\]

    Since $\guess(t_0,\alpha)$ is defined to be the value of $\ell$ which satisfies \eqref{eq:whittle}, 
  this gives the proof.
\end{proof}

\subsubsection{Simplifying the guessed Whittle Index}
\label{sec:simpli}
Lemma~\ref{lemma:solving whittle} gives a complicated expression for the guessed Whittle Index. 
In this subsection, we work on simplifying Equation \eqref{eq:whittle alpha}. The final result is an elegant formula for the guessed Whittle Index (see Lemma~\ref{thm:main}). 
Since the proof is long and intricate, for the readers' convenience, we give a table summarizing each lemma below (Table~\ref{table:lemma}). 

{\small 
\begin{table}[h]
  \caption{Lemmas in Section~\ref{sec:simpli}}
  \label{tab:freq}
  \begin{tabular}{ccl}
    \toprule
    Lemma number  & Summary\\
    \midrule
    Lemma~\ref{lemma:ODE} & A preliminary lemma on ordinary differential equations (ODE)\\
    Lemma~\ref{lemma:exp formula} & A basic equation for exponential variables \\
    Lemma~\ref{lemma:c converge} & A necessary condition for the holding cost to converge \\
    Lemma~\ref{lemma: r'} & Characterizing the cost function $r$\\
    Lemma~\ref{lemma:eq1} & Characterizing $\costbar$\\
    Lemma~\ref{lemma:eq2} & Characterizing $\Gamma$\\
    Lemma~\ref{lemma: four eq} & Giving four equations using Lemma~\ref{lemma:exp formula} and Lemma~\ref{lemma: r'}\\
  \bottomrule
\end{tabular}
\label{table:lemma}
\end{table}
}

The following lemma is a basic formula in ODE. We present it without proof.
\begin{lemma}[First Order ODE]
\label{lemma:ODE}
    For the first order ODE,
    $y'+P(x)y = Q(x),$
    the general solution has the form
    \[y = e^{-\int P(x)\,dx}\left(\int e^{\int P(x)\,dx}Q(x)\,dx + C\right).\]
\end{lemma}

We also give a basic formula on exponential variables.
\begin{lemma}
\label{lemma:exp formula}
    For any smooth function $f$ and exponential variable $X$, we have that 
    \[\E{f(x_0+X)}-f(x_0)=\E{X}\E{f'(x_0+X)}.\]
    \[f(x_0)-\E{f(x_0-X)}=\E{X}\E{f'(x_0-X)}.\]
\end{lemma}
\begin{proof}
    Suppose $X\sim\Exp(\theta)$. Then we have that 
    \begin{align*}
        \E{f(x_0+X)} -f(x_0) = \int_{0}^\infty f'(x_0+t)\P{X>t}dt 
        &=\int_0^\infty f'(x_0+t)e^{-\theta t} dt \\
        &= \frac{1}{\theta} \int_{0}^\infty f'(x_0+t) \theta e^{-\theta t} dt \\
        &= \E{X}\E{f'(x_0+X)}.
    \end{align*}
    The other equation is derived similarly.
\end{proof}

The next lemma is a necessary condition for the TVHC problem to be well-defined. It states that the instantaneous holding cost $c(t)$ cannot grow too fast. This provides the basis for the proof of Lemma~\ref{lemma:eq1}. We present this lemma with the subscripts $i$, but will later drop all subscripts.
\begin{lemma}
\label{lemma:c converge}
    If there exists a policy such that the long-run mean holding cost of the TVHC problem converges, then for any type $i$, we have that 
    \[\lim_{t\to \infty} \frac{\int_0^t c_i(x)dx}{e^{(\mu_i-\lambda_i)t}} = 0.\]
\end{lemma}

\begin{proof}
    If the long-run mean holding cost converges, then it must also be the case that the long-run mean holding cost converges if there are only type $i$ jobs in the system, and all jobs are served in FCFS order.
    In this case, the response time of each type $i$ job  is the response time in M/M/1, which is distributed $\Exp(\mu_i-\lambda_i)$.
    
    Thus we have that the mean holding cost, which can be computed by the product of the arrival rate and the expected accumulated holding cost per job, has the formula
    \[\lambda_i \E{\int_0^A c_i(x)dx},\]
    where $A\sim \Exp(\mu_i-\lambda_i)$.   
    If $\E{\int_0^A c_i(x)dx}$ converges, it is equal to:
    \begin{equation}
        \int_0^\infty (\mu_i-\lambda_i) \frac{\int_0^t c_i(x)dx}{e^{(\mu_i-\lambda_i)t}} dt.
        \label{eq:c converge proof}
    \end{equation}
Thus we know that $\lim_{t\to \infty} \frac{\int_0^t c_i(x)dx}{e^{(\mu_i-\lambda_i)t}} = 0,$ otherwise expression \eqref{eq:c converge proof} does not converge.
\end{proof}

We now return to the single-arm bandit problem. We first need to characterize the cost function $r$ which is defined in \eqref{eq:r2}.

\begin{lemma}[$r'$]
\label{lemma: r'}
The derivative of $r$ is given by
    \[r'(t) = c'(t) + \lambda c(t).\]
\end{lemma}
\begin{proof}
Note that in the definition of $r(s)$ (Equation~\eqref{eq:r2}), $\{Y_j\}$ can be viewed as the set of Poisson arrivals with rate $\lambda$ from time 0 to time $s$. 

 A Poisson arrival process from time 0 to $t+\delta$ can be divided into two independent Poisson arrival processes: the one from time 0 to $t$ and another one from time $t$ to $t+\delta$. Thus we have that 
\begin{equation}
    r(t+\delta) = c(t+\delta) + \E{\sum_{j}c(Y_j)} + \E{\sum_j c(Y_j')},
    \label{eq:r' proof}
\end{equation}
where $\{Y_j\}$ is the Poisson arrivals during $[0,t]$, and $\{Y_j'\}$ is the Poisson arrivals during $[t, t + \delta]$. 
The expected number of arrivals during $[t, t + \delta]$  is $\lambda \delta$, and the instantaneous holding cost of those jobs is approximately $c(t_0)$. Thus
$\E{\sum_j c(Y_j')} = c(t_0)\cdot \lambda\delta + o(\delta)$. This together with \eqref{eq:r' proof} yields:
\[r'(t_0)=\frac{r(t_0+\delta)-r(t_0)}{\delta} = c'(t_0)+\lambda c(t_0).\]
\end{proof}

We now derive the expression of $\costbar(t,\alpha)$, which is defined in Definition~\ref{def:decay costbar}.

\begin{lemma}[Formula for $\costbar$]
\label{lemma:eq1}
For any $t$ and $\alpha>0$,
    \begin{equation}
    \costbar(t,\alpha) = (1-\gamma_1)\E{r(t+X)}, \qquad \mbox{ where } X \sim \Exp\left(\frac{\alpha}{1 - \gamma_1}\right)
    \label{eq: equation1}
\end{equation}
\end{lemma}

\begin{proof}
Recall that $\costbar(t,\alpha)$ is defined to be 
\[\costbar(t,\alpha):=\E{\int_0^{T_1} \alpha r(t+A(s)) e^{-\alpha s} ds},\]
where $A(s)$ is the random variable denoting the age of the oldest job in an M/M/1 queue at time $s$. 

During the first $\delta$ time period, a job completes with probability $\mu \delta$. If this happens, the state $A(s)$ drops by an exponential amount ($A(\delta) = \delta-\Exp(\lambda)$). There are two cases: (1) with probability $e^{-\lambda \delta}$, the decrement is more than $\delta$, which means $\A(\delta)<0$ and the busy period ends; (2) with probability $1-e^{-\lambda \delta}$, the dropping amount is smaller than $\delta$ and the busy period continues. In this case, suppose $A(\delta)=\Delta<\delta$. Thus, conditioning on a drop happening in the first $\delta$ time, we have  that $\costbar$ is, with a slight abuse of notation: 
\begin{equation}
\left [\costbar(t,\alpha) \mid \text{a drop happens in the first $\delta$ time}\right ] = \alpha r(t) \delta + \left(1-e^{-\lambda \delta}\right) e^{-\alpha \delta} \costbar(t+\Delta,\alpha)+o(\delta).
\label{eq:costbar1}
\end{equation}

Otherwise, with probability $1-\mu\delta$, no drop happens during the first $\delta$ time. 
Then $A(s)$ starts at $\delta$, and experiences an M/M/1 busy period like process until it drops below $\delta$. During this time, $\costbar(t+\delta,\alpha)$ is incurred.
At the point that $A(s)$ drops below $\delta$, we know that $A(s)$ is distributed as $\delta - \Exp(\lambda)$ because this is an exponential overshoot. Thus
with probability $e^{-\lambda \delta}$, $A(s)$ is below $0$, which means the busy period ends. Otherwise with probability $1-e^{-\lambda \delta}$, $A(s)$ is still larger than 0 (again, denoted by $\Delta$) and the busy period continues.  Thus 
the conditional value of $\costbar$ in this case is
\begin{align}
&\left [\costbar(t,\alpha) \mid \text{no drops happen in the first $\delta$ time}\right ]\nonumber\\ 
&= \alpha r(t) \delta + e^{-\alpha \delta} \left(\costbar(t+\delta,\alpha) + \E{e^{-\alpha T_1}}(1-e^{-\lambda \delta}) \costbar(t+\Delta,\alpha)\right)+o(\delta).  
\label{eq:costbar2}
\end{align}

Combining \eqref{eq:costbar1} and \eqref{eq:costbar2}, we have that 
\begin{align*}
    \costbar(t,\alpha) =&  \alpha r(t) \delta + \mu\delta\left(\left(1-e^{-\lambda \delta}\right) e^{-\alpha \delta} \costbar(t+\Delta,\alpha)\right) +o(\delta)\\
    &+ (1-\mu\delta)\left(e^{-\alpha \delta} \left(\costbar(t+\delta,\alpha) + \gamma_1(1-e^{-\lambda \delta}) \costbar(t+\Delta,\alpha)\right)\right).
\end{align*}

Thus we have that 
\begin{align*}
    &\lim_{\delta \to 0} \frac{\costbar(t+\delta,\alpha)-\costbar(t,\alpha)}{\delta}  \\
    = & -\alpha r(t) - \lim_{\delta\to 0}\mu\left(\left(1-e^{-\lambda \delta}\right) e^{-\alpha \delta} \costbar(t+\Delta,\alpha)\right)\\
    &+ \lim_{\delta\to 0}  \frac{1-(1-\mu\delta)e^{-\alpha \delta}}{\delta} \costbar(t+\delta,\alpha) - \lim_{\delta\to 0} (1-\mu\delta)e^{-\alpha \delta}\gamma_1\frac{1-e^{-\lambda\delta}}{\delta} \costbar(t+\Delta,\alpha)\\
    =&  -\alpha r(t) + (\mu+\alpha)\costbar(t,\alpha) - \lambda\gamma_1 \costbar(t,\alpha)\\
    =& -\alpha r(t) + (\mu+\alpha-\lambda\gamma_1)\costbar(t,\alpha).
\end{align*}

Note that this is a first-order ODE of the function $\costbar(t,\alpha)$ with respect to $t$. By Lemma~\ref{lemma:ODE}, 
\begin{equation}
    \costbar(t,\alpha) = e^{(\mu+\alpha-\lambda\gamma_1)t} \left(\int -e^{-(\mu+\alpha-\lambda\gamma_1)s} \alpha r(s)ds+C \right).
    \label{eq:costbar proof eq}
\end{equation}

Note that by \eqref{eq:gamma2}, 
\[\alpha+\mu-\lambda\gamma_1 = \alpha - \lambda\gamma_1 + \frac{(\alpha-\lambda\gamma_1+\lambda)\gamma_1}{1-\gamma_1} = \frac{\alpha}{1-\gamma_1}.\]

Thus \eqref{eq:costbar proof eq} can be rewritten as

\begin{align}
    \costbar(t,\alpha) &= e^{\frac{\alpha}{1-\gamma_1}t} (\int_t^\infty e^{-\frac{\alpha}{1-\gamma_1}s} \alpha r(s)ds+C) \nonumber\\
    &= \int_0^\infty e^{-\frac{\alpha}{1-\gamma_1}s}\alpha r(t+s) ds + Ce^{\frac{\alpha}{1-\gamma_1}t}\nonumber\\
    &= (1-\gamma_1)\E{r(t+X)} + Ce^{\frac{\alpha}{1-\gamma_1}t},
    \label{eq: costbar with C}
\end{align}
where $X\sim \Exp(\frac{\alpha}{1-\gamma_1})$. Next we prove that $C=0$.

By Lemma~\ref{lemma: r'}, we have that 
\[r(t)=c(t)+\lambda \int_0^tc(x)dx + C_1.\]
    Together with Lemma~\ref{lemma:c converge} (dropping the subscripts $i$), we have that 
    \[\lim_{t\to 0} \frac{r(t)}{e^{(\mu-\lambda)t}} = 0.\]

    From ~\eqref{eq: costbar with C} we have that for any $t$,
    \[|C|\leq \costbar(t,\alpha)e^{-\frac{\alpha}{1-\gamma_1}t} + (1-\gamma_1)\E{r(t+X)}e^{-\frac{\alpha}{1-\gamma_1}t}.\]
    However, both terms on the right hand side go to zero as $t\to \infty$ 
    (see Lemma~\ref{lemma: app for cost bar} in Appendix~\ref{app:defer proofs}).
    Thus we have proven that $C=0$. By Equation \eqref{eq: costbar with C} we have the proof.
\end{proof}

We now characterize another important term in \eqref{eq:whittle alpha}, the function $\Gamma$.
\begin{lemma}[Formula for $\Gamma$]
\label{lemma:eq2}
    \begin{equation}
    \label{eq: equation2}
        \Gamma(t,\alpha)= \frac{\alpha}{\alpha+\lambda}\E{r(t-T_2)}.
    \end{equation}
\end{lemma}
\begin{proof} The proof is relatively straightforward:
    \begin{align}
    \Gamma(t,\alpha):=\E{\int_0^{T_2}\alpha r(t-T_2+s)e^{-\alpha s}ds} 
    &=\int_{0}^\infty \int_0^x \alpha r(t-x+s)e^{-\alpha s} ds\  \lambda e^{-\lambda x}dx \nonumber\\
    &= \int_{0}^\infty \int_s^\infty \alpha r(t-x+s)e^{-\alpha s}   \lambda e^{-\lambda x}dxds \nonumber\\
    &= \int_{0}^\infty \int_0^\infty \alpha r(t-y)e^{-\alpha s}   \lambda e^{-\lambda (y+s)}dyds \nonumber\\
    &= \int_{0}^\infty r(t-y)\lambda e^{-\lambda y} dy \cdot \int_{0}^\infty \alpha  e^{-(\alpha+\lambda)s} ds\nonumber\\
    &= \frac{\alpha}{\alpha+\lambda}\E{r(t-T_2)}.\nonumber
\end{align}
\end{proof}

Finally, we use the following lemma to simplify \eqref{eq:whittle alpha}.  
\begin{lemma}
\label{lemma: four eq}
Let $T_2\sim \Exp(\lambda)$, $X\sim \Exp(\frac{\alpha}{1-\gamma_1})$. Then we have the following equations:
\begin{align}
    &\E{r'(t-T_2)}=\lambda c(t).
    \label{eq: equation3}\\
    &\E{r(t-T_2)}= r(t)-c(t).
    \label{eq: equation4}\\
    &\E{r'(t+X)}= \frac{\mu}{\gamma_1}\E{c(t+X)} - \frac{\alpha}{1-\gamma_1}c(t).
    \label{eq: equation5}\\
    &\E{r(t+X)}=r(t)-c(t)+\frac{\mu(1-\gamma_1)}{\gamma_1\alpha}\E{c(t+X)}.
    \label{eq: equation6}
\end{align}
    


    
\end{lemma}
\begin{proof}
The equations follow from Lemmas~\ref{lemma:exp formula} and~\ref{lemma: r'}. See Lemma~\ref{lemma: four eq app} in  Appendix~\ref{app:defer proofs}.
\end{proof}

Now substituting \eqref{eq: equation1}, \eqref{eq: equation2}, \eqref{eq: equation3}, \eqref{eq: equation4}, \eqref{eq: equation5}, \eqref{eq: equation6} into the formula given in Lemma~\ref{lemma:solving whittle} gives the desired simple guessed Whittle Index. 

\begin{lemma}
\label{thm:main}
    The guessed Whittle Index given discount factor $\alpha$ is:
    \begin{equation}
    \label{eq:guess W}
        \guess(t_0,\alpha) = \mu \E{c(t_0 + X)},\quad  \text{where }X\sim \Exp(\frac{\alpha}{1-\gamma_1}).
    \end{equation}
\end{lemma}

\subsection{Indexability and the Whittle Index}
\label{sec:index}

We now use the guessed Whittle Index given in Lemma~\ref{thm:main} to prove indexability and the correctness of the guess. We only present the proof sketch here due to space limitations and defer the proof to the appendix. However, we need to point out that, though deferred, the proof is also involved.
\begin{customthm}[Indexability]
\label{thm:index}
    The discounted single-arm bandit problem is indexable with the Whittle Index $W(t_0,\alpha)=\guess(t_0,\alpha)$, where $\guess$ is the guessed Whittle Index given in Lemma~\ref{thm:main}.
\end{customthm}

\begin{proof}
    The main lemma for this theorem is to prove that when the compensation $\ell=\guess(t_0,\alpha)$, the policy $Threshold(t_0)$ is the optimal policy. This is proved by verifying the Hamilton–Jacobi–Bellman equation. See Appendix~\ref{app:index} for the detailed proof. 
\end{proof}

Our final result is the Whittle Index for the scheduling problem. It is a straightforward corollary of Theorem~\ref{thm:index} by taking the limit on the discount factor and adding back the subscript $i$. 

\begin{customcor}
\label{cor:main}
    The Whittle Index for the scheduling problem is 
    \[W_i(t_0) = \mu\E{c_i(t_0+X)}, \text{ where }X\sim \Exp(\mu_i-\lambda_i).\]
\end{customcor}
\begin{proof}
    The formula is given by $W(t_0) = \lim_{\alpha\to 0}\guess(t_0,\alpha)$:
    \[\lim_{\alpha\to 0} \frac{\alpha
        }{1-\gamma_1} = \lim_{\alpha\to 0} \frac{\alpha
        }{1-\widetilde{T_1}(\alpha)} = \frac{1}{-\widetilde{T_1}'(0)} = \frac{1}{\E{T_1}} = \mu-\lambda.\]
    Adding back the subscript $i$ yields the Whittle Index. 
\end{proof}

\section{Simulations}
\label{sec:simulation}

We conduct simulations to evaluate the performance of our proposed policy (from Corollary~\ref{cor:main}), comparing it against other policies in the literature.  In our evaluations, we experiment with different numbers of classes, different holding cost functions, and a range of system loads. We present the following main findings:

\begin{itemize}
    \item Across all experiments, including cases with complex interleaving holding cost functions, our policy consistently matches the lowest time-average holding cost among all candidate policies we considered. 
    \item In certain problem settings, our policy significantly outperforms each candidate policy. 
\end{itemize}


\subsection{Policies evaluated}
\label{s:policies}
Throughout, when we talk about \enquote{our policy,} we will refer to the policy from Corollary~\ref{cor:main} where the priority of a class $i$ job of age $t$ is 
\[
    W_i(t) = \mu_i\E{c_i(t+X)}, \qquad \text{ where }X\sim\Exp(\mu_i-\lambda_i)
\]
We compare our policy with a broad spectrum of alternative policies:
\begin{description}
    \item[FCFS.] This policy always serves the job that arrived earliest.  FCFS is a very simple policy and we compare against it as a baseline. 
    \item[Strict Preemptive Priority.] 
    This policy assigns a fixed priority to each job class, where jobs from a higher-priority class have preemptive priority over those from a lower-priority class.  Jobs within a class are run in FCFS order.
    We choose a highly optimistic version of this policy where the priority ordering can change at each value of system load. Thus we might run Prio(1;2), where class 1 has priority over class 2 when load is low, but Prio(2;1) when load is high.
    
    \item[Generalized $c\mu$~\cite{van1995dynamic}.] This policy always serves the job with highest index $c_i(t)\cdot\mu_i$, where $t$ is the age of the class $i$ job. 
    This policy is known to be optimal in the diffusion limit. 
    \item[Aalto's Index~\cite{aalto2024whittle}.] 
    In this policy a class $i$ job of age $t$ is given index $V_i(t)$, where\footnote{This is a simplification of Aalto's policy to the case of exponential job sizes.}
    \[
        V_i(t) = \mu_i\E{c_i(t+S)},\qquad\text{ where }S\sim Exp(\mu_i).
    \]
    Like our policy, this is also a Whittle-index based heuristic, but is motivated by the static setting (no arrivals), and hence does not incorporate the arrival rate $\lambda_i$.  
    
\end{description}

\subsection{Experimental Results}

We have conducted hundreds of experiments comparing the scheduling policies from Section~\ref{s:policies} under  different holding cost functions, job sizes, and arrival rates.  In all of these, our policy was either the best policy or matched the best policy.  Due to lack of space, we present only four diverse experiments which well-illustrate some important behaviors.


Each experiment shown is represented by  {\em (a)} a set of holding cost functions (not drawn to scale) and {\em (b)} the corresponding simulation results. We experiment with small and large jobs, where holding cost functions drawn in red or orange correspond to a class with shorter jobs, while those drawn in blue correspond to a class with larger jobs.
We experiment with different arrival rates, shown as a proportion of the total arrival rate, $\lambda$, across all classes, where $\lambda$ is specified by the load.


\Cref{fig:1-deadline-cost-fn} considers an experiment where   class 1 jobs (the shorter ones) only incur a holding cost after a deadline is passed, while class 2 jobs (the longer ones) incur a steady (but low) holding cost.  In this figure, the arrival rate of class 1 is much higher than that of class 2.  
\Cref{fig:1-deadline-plot} shows the corresponding results.  We see that FCFS is by far the worst policy.  The generalized $c \mu$ rule improves upon FCFS, and Aalto improves upon that.  Our policy (shown in red) is significantly better than all the others, except for Preemptive Priority, which is equal to our policy here. 

To understand what's going on, we first observe that we should be prioritizing class 1 jobs ahead of their deadline, given that the cost of missing the deadline is so high.  
The $c \mu$ rule only prioritizes class 1 at the deadline point.   The Aalto policy improves upon the $c \mu$ rule, by prioritizing class 1 jobs in advance of the deadline, but it doesn't do so early enough (because Aalto doesn't consider load).  Our policy improves upon the Aalto policy by prioritizing class 1 jobs even earlier -- in fact a better time to start prioritizing class 1 jobs is at age 0, which is what the Preemptive Priority policy does.


\begin{figure}[htp]
    \centering
    \begin{subfigure}[b]{0.40\textwidth}
        \centering
        \includegraphics[width=\textwidth]{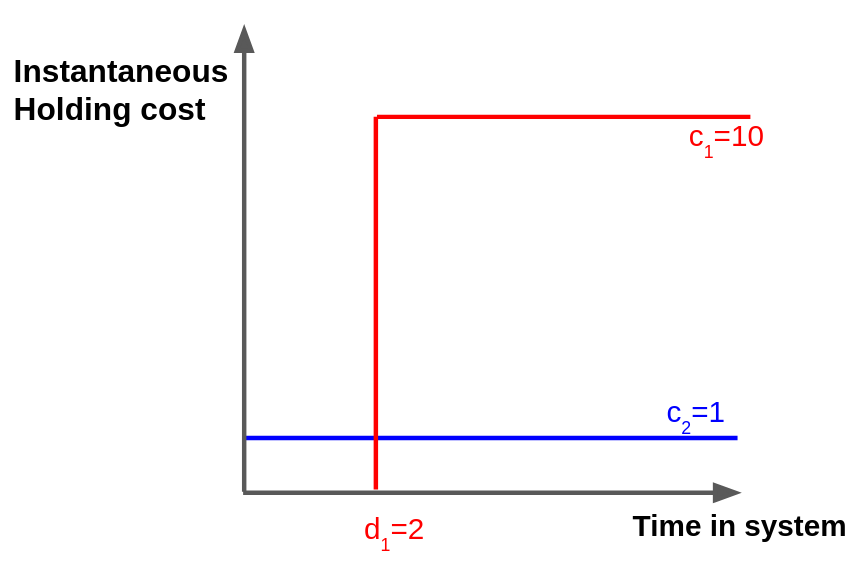}
        \caption{Holding cost functions.} 
        \label{fig:1-deadline-cost-fn}
    \end{subfigure}
    \hfill
    \begin{subfigure}[b]{0.40\textwidth}
        \centering
        \includegraphics[width=\textwidth]{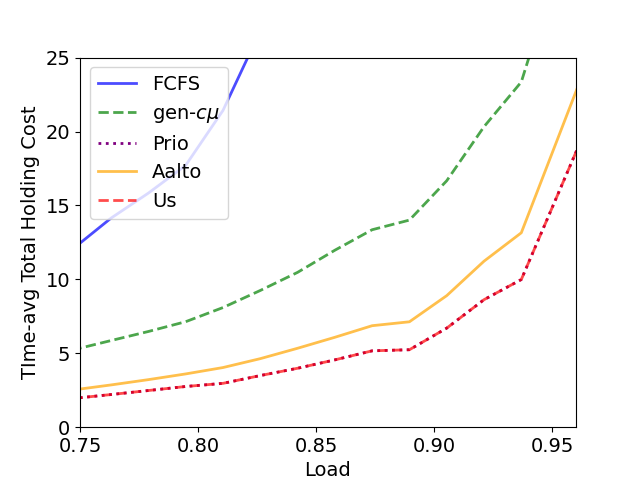}
                \caption{Performance of policies.}
        \label{fig:1-deadline-plot}
    \end{subfigure}
    \caption{Comparison of policies on holding cost functions with one deadline. We fix $\mu_1=3, \mu_2=1, \lambda_1=0.9\lambda$. }
    \label{fig:1-deadline}
\end{figure}

\Cref{fig:2-deadline-balanced-cost-fn} considers an experiment where both class 1 jobs (the short ones) and class 2 jobs (the long ones) only incur holding costs after their respective deadlines are passed. Here class 1 jobs incur a high penalty after a late deadline, while class 2 jobs incur a low penalty after an early deadline. The two classes have equal arrival rates. 
\Cref{fig:2-deadline-balanced-plot} shows the corresponding results.  
We see that our policy is noticeably better than all the others, except for Aalto's policy, which we only negligibly dominate. Our policy noticeably dominates Preemptive-Priority, which noticeably dominates generalized $c \mu$ which significantly dominates FCFS.


To understand what's going on, observe that it is again useful to prioritize jobs {\em ahead} of their deadlines, which explains why we outperform FCFS and generalized $c\mu$.
Strict priority is no longer effective here, because the deadlines are further out, so always prioritizing the class 1 jobs can be suboptimal. Our policy's performance is similar to Aalto for two reasons:  First, because the arrival rates are balanced, we find that when load is not too high, $\lambda_i$ is small compared to $\mu_i$, and our policy is similar to Aalto's policy.  Second, while our policy looks further into the future, it does so for both classes, and hence the benefit in some sense ``cancels out,'' leaving us with a policy similar to Aalto.  


\begin{figure}[htp]
    \centering
    \begin{subfigure}[b]{0.40\textwidth}
        \centering
        \includegraphics[width=\textwidth]{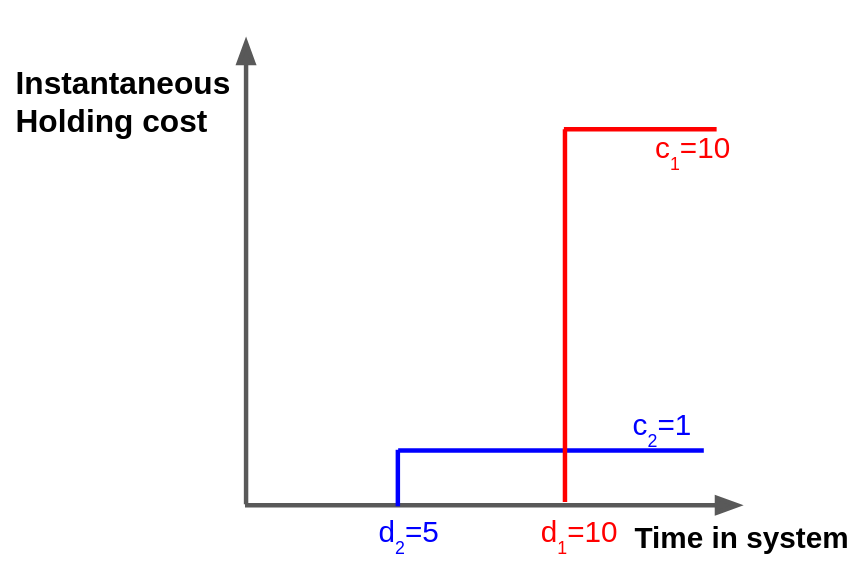}
        \caption{Holding cost functions.} 
        \label{fig:2-deadline-balanced-cost-fn}
    \end{subfigure}
    \hfill
    \begin{subfigure}[b]{0.40\textwidth}
        \centering
        \includegraphics[width=\textwidth]{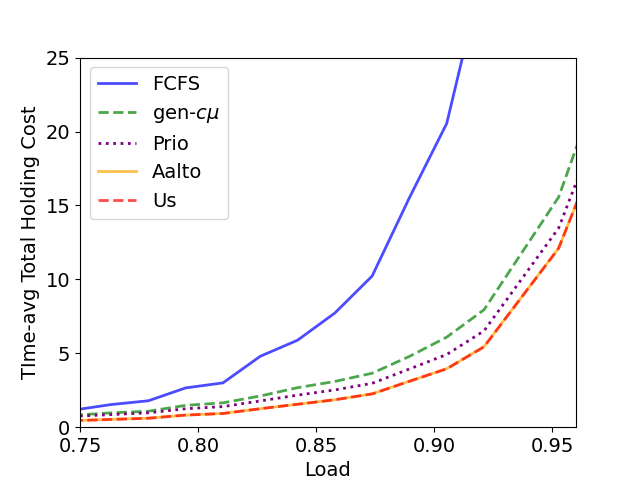}
        \caption{Performance of policies.} 
        \label{fig:2-deadline-balanced-plot}
    \end{subfigure}    
    \caption{Comparison of policies on holding cost functions with deadlines. We fix $\mu_1=3, \mu_2=1, \lambda_1=0.5 \lambda$.}
    \label{fig:2-deadline-balanced}
\end{figure}


\Cref{fig:polynomial-cost-fn} considers an experiment where jobs have polynomial holding cost functions. Specifically, class 1 jobs (the shorter ones) have a linear holding cost, while class 2 jobs (the longer ones) have a quadratic holding cost. We set the arrival rates such that the loads from both classes are balanced. Figure~\ref{fig:polynomial-cost-plot} shows the corresponding results. We see that our policy significantly outperforms FCFS and Preemptive-Priority; however, our policy only negligibly beats generalized $c\mu$ and Aalto.


To understand why the three dominant policies (our policy, generalized $c\mu$ and Aalto) are similar, we note that, for the polynomial holding cost functions, these three policies have index functions which are polynomials with the same leading coefficient. Consequently, their behaviors are similar. We often find  that these three policies perform similarly and vastly improve upon the other policies.


\begin{figure}
    \centering
    \begin{subfigure}[b]{0.40\textwidth}
        \centering
        \includegraphics[width=\textwidth]{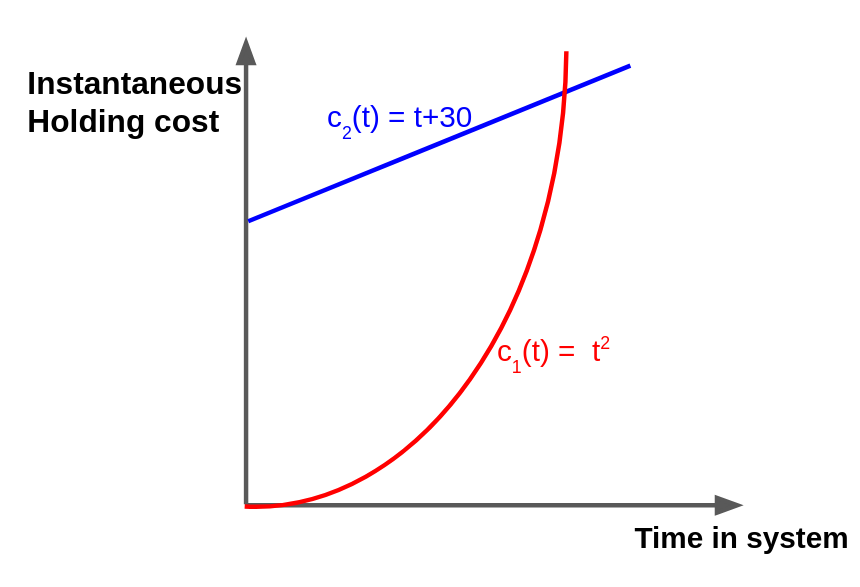}
        \caption{Holding cost functions.}
        \label{fig:polynomial-cost-fn}
    \end{subfigure}
    \hfill
    \begin{subfigure}[b]{0.40\textwidth}
        \centering
        \includegraphics[width=\textwidth]{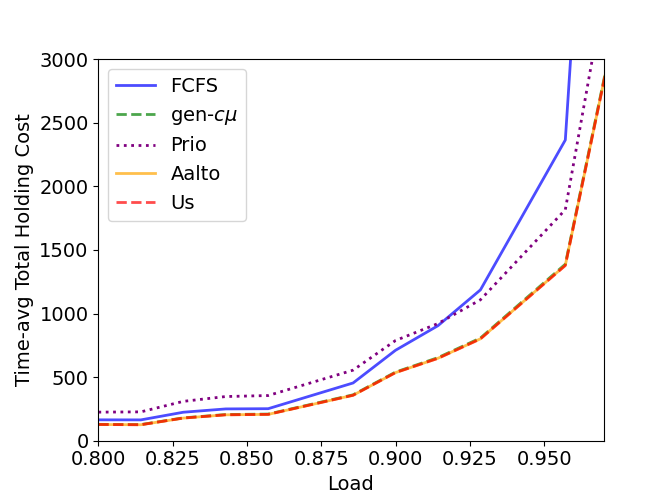}
        \caption{Performance of policies.}
        \label{fig:polynomial-cost-plot}
    \end{subfigure}    

    \caption{Comparison of policies on polynomial holding cost functions. We fix $\mu_1=3, \mu_2=1, \lambda_1 = 0.75\lambda$.}
    \label{fig:polynomial-cost}
\end{figure}
Lastly, \Cref{fig:multiclass-cost-fn} considers the case with three job classes. Specifically, class 1 and class 2 jobs (shorter jobs) have linear interleaving holding cost functions, while class 3 jobs (longer jobs) have explosive quadratic holding cost. Figure~\ref{fig:multiclass-cost-plot} shows the corresponding results. We see that our policy noticeably improves upon Aalto, which noticeably improves upon generalized $c \mu$, which noticeably improves upon Preemptive-Priority, which significantly improves upon FCFS. 

The strict ordering of policies, depicted in Figure~\ref{fig:multiclass-cost-plot} is typical of what we see in many experiments, regardless of the number of classes. This strict separation would also be more obvious in Figure~\ref{fig:polynomial-cost-fn} if we set both holding costs to be higher.  When the costs are higher, the small differences between the policies are more amplified.


Another policy which we experimented with is the Accumulated Priority policy from \cite{stanford2014waiting,fajardo2017waiting}.  However we found that this policy never improved upon our policy in our experiments and was often significantly worse than our policy, so we chose to omit it to keep the graphs cleaner.


\begin{figure}
    \centering
    \begin{subfigure}[b]{0.40\textwidth}
        \centering
        \includegraphics[width=\textwidth]{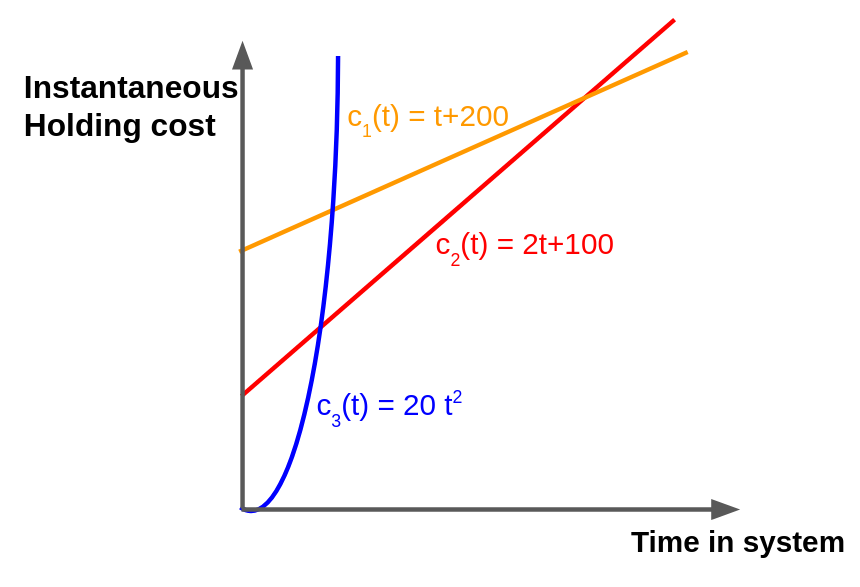}
        \caption{Holding cost functions. }
        \label{fig:multiclass-cost-fn}
    \end{subfigure}
    \hfill
    \begin{subfigure}[b]{0.40\textwidth}
        \centering
        \includegraphics[width=\textwidth]{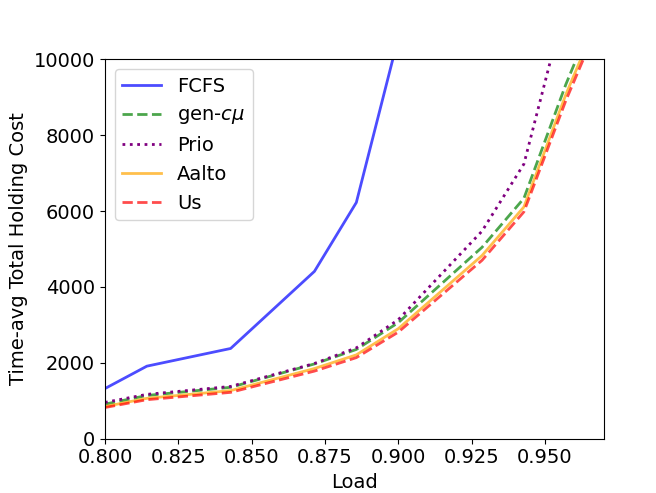}
        \caption{ Performance of policies.}
        \label{fig:multiclass-cost-plot}
    \end{subfigure}
    \caption{Comparison of policies on $3$ classes. We fix $\mu_1=\mu_2=3, \mu_3=1, \lambda_1 = \lambda_2=\lambda_3$.}
    \label{fig:multiclass-cost}
\end{figure}

\section{Discussion on Optimality}

Our policy replicates the diffusion limit optimality of the generalized $c\mu$ rule \cite{van1995dynamic}.  
The diffusion limit regime in \cite{van1995dynamic} is equivalent to $\lambda_i =\Theta(n),\mu_i=\Theta(n),\frac{\lambda_i}{\mu_i}=1-\Theta(\frac{1}{\sqrt{n}})$ where $n\to \infty$. Thus we have that 
\(\mu_i-\lambda_i\to \infty.\)
This indicates that our policy degenerates to the generalized $c\mu$ rule, and hence is optimal as well.

We have seen that our policy performs well in simulation. However, like all the other Whittle-based heuristics we've discussed, our policy is not always optimal. We now give a counter example.

Assume that there are two classes of jobs, both with {\em the same} instantaneous holding cost function $c(t)=t$. Both classes have the same completion rate, but the arrival rates are different. Since the holding cost functions and the completion rates are the same in both classes, the optimal policy is FCFS, which is clearly not our policy.  
Thus, while our Whittle Index policy performs really well compared with other existing policies, there is still a need for further work on optimality. 
\section{Conclusion}

This paper studies the classical TVHC problem: Jobs of different classes arrive over time, where each class of jobs is associated with a holding cost that increases with the job's age. The objective is to schedule the jobs so as to minimize the expected time-average total holding cost. 

Various papers have provided heuristics for the TVHC problem.
The generalized $c \mu$ rule \cite{van1995dynamic} provides a simple index policy which favors jobs with currently high holding cost and high failure rate; this policy is asymptotically optimal in the diffusion limit.  More recent works by Aalto \cite{aalto2024whittle}, consider the more tractable static setting (no arrivals).

Our work, while also heuristic, takes a principled approach to the TVHC problem with arrivals: (i) We derive the first representation of our problem as an R-MAB with a {\em finite} number of arms, and (ii) We derive a novel Whittle index policy for the resulting R-MAB.  While the analysis is involved, the resulting policy is extremely simple and elegant and incorporates class load.  In simulation, our policy improves upon all the other known heuristics.

This story is in no way finished.  First, our policy might  be generalizable beyond exponential job size distributions, to those with increasing hazard rate, using our existing R-MAB as a foundation (although this would require a more complex state space).  Second, there is much more work needed to find an {\em optimal} policy.

{\scriptsize

\bibliographystyle{abbrv}
\bibliography{bib}
}

\appendix

\section{Tutorial on The Whittle Index}
\label{sec:2.2}

The Whittle Index approach was introduced by Whittle in 1988 to develop a heuristic for R-MAB problems (\cite{whittle1988restless}), especially for discounted R-MAB problems. There are recent works proving asymptotic optimality of the Whittle Index (e.g., \cite{ayesta2021computation}); even when not optimal, the Whittle Index has long been shown to be a surprisingly good heuristic (see \cite{nino2023markovian} for a recent review).

The standard Whittle approach (for discounted R-MAB problems) consists of three steps:
\begin{enumerate}
    \item First, the hard constraint that only one arm can be pulled at a time is relaxed to a time-average constraint. Mathematically, suppose the action for arm $i$ at time step $t$ is $a_i(t)$, where $a_i(t)=1$ if the action is active and $a_i(t)=0$ otherwise. The hard constraint is \[\sum_i a_i(t)\leq 1,\ \forall t.\] Whittle relaxes the hard constraint to the following time-average constraint:
    \[\sum_{t=0}^\infty \alpha^t \sum_{i=1}^k a_i(t) \leq \frac{1}{1-\alpha},\]
    where $\alpha$ is the discount factor.
    Through this relaxation, the optimization problem becomes more tractable.  The relaxed R-MAB is studied to obtain a heuristic for the original R-MAB.

    \item Next, a Lagrange method is applied to the relaxed R-MAB. Crucially the Lagrange dual function can be decomposed into $k$ independent parts, and accordingly the relaxed R-MAB problem can be decomposed to $k$ independent {\em single-arm restless bandit problems}, where not pulling the arm earns a compensation of $\ell$: $\ell$ serves as the Lagrange multiplier. This step effectively puts a ``price'' on choosing to keep an arm active, yielding $k$ single-arm optimization problems, each parameterized by $\ell$.

    \item Then, one solves each single-arm problem to figure out the critical value of $\ell$ for each state of each problem, which is the exact point where one is indifferent between being active and passive. This critical value of $\ell$ measures the ``value'' of pulling this arm at each state, and serves as the Whittle Index. 

    \item Finally, the Whittle Index policy pulls the arm with the largest Whittle Index at every moment of time.
\end{enumerate}

In the literature, people usually skip the first two steps and directly analyze the single-arm bandit problem (e.g., \cite{anand2018whittle, aalto2024whittle}.) In this paper, we will likewise omit the first two steps and directly provide the single-arm bandit formulation in Section~\ref{sec:single}.

A property of the single-arm bandit problem called ``indexability'' needs to be established for the Whittle Index to be well-defined. Intuitively, indexability states that as the compensation $\ell$ grows larger, the optimal action at a certain state can only change from active to passive. This property ensures that the Whittle Index is unique and well-defined. Although indexability is intuitive, it must be proven~\cite{whittle1988restless}.

\section{Proof for Lemma~\ref{lemma:FCFS}}
\label{app:lemma3.1}
\begin{lemma}[FCFS within each type is optimal]
    The optimal policy must serve jobs within each type in FCFS order. 
\end{lemma}

\begin{proof}

For any policy $\pi$, we define a policy $\pi_{FCFS}$ where $\pi_{FCFS}$ serves a type $i$ job if and only if  $\pi$ serves a type $i$ job, but $\pi_{FCFS}$ always serves FCFS within each type.  We will show that the time-average total holding cost under $\pi_{FCFS}$ is no higher than that of $\pi$ on every sample path.

We now show how to define the sample path that allows us to compare between $\pi$ and $\pi_{FCFS}$.  It is important to observe that the result of this lemma does not hold for the obvious sample path, which is based on job arrival times and sizes (think about SRPT within each class, which can improve upon FCFS within each class).  Thus we need to define our sample path to leverage the fact that job sizes are memoryless.  

We define the sample path as $2$ sequences of exponential random variables for each of the $k$ classes.  The first sequence for class $i$ represents inter-arrival times of class $i$ jobs.  These are clearly distributed as $\Exp(\lambda_i)$.  The second sequence for class $i$ represents inter-departure times of class $i$ jobs, where we only look at the system when a class $i$ job is being served.   One can imagine that there is an $\Exp(\mu_i)$ timer for class $i$ which is {\em paused} when the system is working on a job of some type other than class $i$, and {\em resumes} when the system returns to working on a class $i$ job. Whenever the timer goes off, the class $i$ job currently being served completes. Note that there might be multiple type $i$ jobs running during the time before the timer goes off.




We now consider the behavior of policies $\pi$ and $\pi_{FCFS}$ under a fixed sample path that consists of the $2k$ exponential sequences. Note that $\pi_{FCFS}$ works on a type $i$ job if and only if $\pi$ works on a type $i$ job. Thus the exponential timers go off at the same time under both policies. Suppose the timer for class $i$ goes off at times $\{d_1
< d_2 <...\}$, at which times a type $i$ job completes and departs. Also let the $j^{th}$ class $i$ job's arrival time be denoted by $a_j$.

Given that $\pi_{FCFS}$ serves type $i$ jobs in FCFS order, we have that the $j^{th}$ type $i$ job leaves at time $d_j$. Define $p(j)$ to be a permutation such that $d_{p(j)}$ is the departure time of the $j^{th}$ type $i$ job under $\pi$. Now we look at the total holding cost incurred by policies $\pi$ and $\pi_{FCFS}$.


For any type $i$, the total holding cost of $\pi_{FCFS}$ incurred by the first $N$ type $i$ jobs is 
\begin{equation}
    \sum_{j=1}^N \int_{0}^{d_j-a_j} c_i(t) dt, 
    \label{eq:FCFS eq1}
\end{equation}

while that of $\pi$ is 
\begin{equation}
    \sum_{j=1}^N \int_{0}^{d_{p(j)}-a_j} c_i(t) dt.
    \label{eq:FCFS eq2}
\end{equation}

Observe that (\ref{eq:FCFS eq1}) is smaller than (\ref{eq:FCFS eq2}) by the monotonocity and convexity of the function $\int_0^x c_i(t) dt$ with respect to $x$.  This completes the proof.
\end{proof}

\section{Proof for Theorem~\ref{thm:bandit}}
\label{app:thm3.2}

\begin{theorem}[Equivalence of TVHC and R-MAB Problems]
For any set of non-decreasing index functions $\{V_i(\cdot)\}_{i=1}^k$, the corresponding index policies (breaking ties by FCFS) in the TVHC problem and the R-MAB problem incur the same cost.
\end{theorem}
\begin{proof}
In this appendix, 
    we construct a coupling between the TVHC problem and the R-MAB problem on a certain sample path. 

    Define a sample path to be $2k$ sequences of exponential variables: $\{a_{i0},a_{i1},a_{i2},...\}$ where $a_{ij}\sim \Exp(\lambda_i)$ and $\{s_{i1},s_{i2},...\}$ where $s_{ij}\sim \Exp(\mu_i)$. Under this sample path, the index policies $\{V_i\}$ run as follows:
    \begin{enumerate}
        \item \textit{The TVHC Problem:} The first type $i$ job enters at time $a_{i0}$. The inter-arrival time of type $i$ jobs are $a_{ij}$. The  $j^{th}$ type $i$ job has size  $s_{ij}$.

        \item \textit{The R-MAB Problem:} The arm $i$ is initialized at state $-a_{i0}$. After the arm is active for $s_{ij}$ time, it triggers the $j^{th}$ state drop, and the dropping amount is $a_{ij}$.
    \end{enumerate}

    Under a fixed sample path, define $T_i(t)$ to be the state of the arm $i$ in the bandit problem, and define $A_i(t)$ to be the age of the oldest type $i$ job in the system in the scheduling problem (if there is no type $i$ job, $A_i(t)<0$ and $-A_i(t)$ is the time until the next arrival of a type $i$ job). We next prove that $T_i(t)=A_i(t)$ for all $t$.

    To prove this fact, we define $\xi_i(t)$ in the bandit problem to be the total time that arm $i$ is active since the last time the state of arm $i$ dropped. We also define a counter, $count_i^B(t)$, to be the number of state drops of arm $i$ by time $t$. We define $\tau_i(t)$ in the scheduling problem to be the total service time that the oldest type $i$ has received at time $t$. We also define a counter, $count_i^S(t)$, to be the number of type $i$ job completions by time $t$. 
    
    Then in the bandit problem, the triple $(T_i(t),\xi_i(t), count_i^B(t))$ is initialized to be $(-a_{i0},0,0)$ and has the following transition function at time $t$:
    \begin{align}
        &\text{If }  i=\arg \max_j \{V_j(T_j(t))\}\text{ and }T_i(t)\geq 0:\nonumber\\
        &\quad\text{If } \xi_i(t)<s_{i,count_i^B(t)}: \nonumber\\
        & \qquad T_i(t+\delta) = T_i(t) + \delta \nonumber\\
        &\qquad\xi_i(t+\delta) = \xi_i(t)+\delta \nonumber\\
        &\qquad count_i^B(t+\delta) = count_i^B(t)\nonumber\\
        &\quad \text{Else:} \nonumber\\
        &\qquad T_i(t+\delta) = T_i(t) + \delta - a_{i, count_i^B(t)} \nonumber\\
        &\qquad\xi_i(t+\delta) = 0 \nonumber\\
        &\qquad count_i^B(t+\delta) = count_i^B(t)+1\nonumber\\
        &\text{Else:} \nonumber\\
        &\quad T_i(t+\delta) = T_i(t) + \delta \nonumber\\
        &\quad\xi_i(t+\delta) = \xi_i(t) \nonumber\\
        &\quad count_i^B(t+\delta) = count_i^B(t)\nonumber\\
        \label{eq:trans of bandit}
    \end{align}

    While we're in the TVHC problem, the triple $(A_i(t),\tau_i(t),count_i^S(t))$ follows the same exact transition function. In this way we prove that $A_i(t)=T_i(t)$ under a fixed sample path.

    Now we prove that the expected cost under the two problems is the same. For the TVHC problem, the long-run average holding cost of type $i$ is the expected total type $i$ holding cost at a uniformly random time $t^*$, which is 
    \[\sum_{\text{$j$: type $i$ job in the system at $t^*$}}c_i(a(j)),\]
    where $a(j)$ is the age of the job $j$. Denote the oldest age among all type $i$ jobs in the system by $A_i(t^*)$. Since the policy runs FCFS within each type, the other type $i$ jobs' inter-arrival times are independent of $A_i(t^*)$, which means $\{a(j): j \text{ is a type $i$ job in the system}\}$ is distributed as $Pois_{\lambda_i}(A_i(t^*))$.
    Thus we have that 
    \[\E{\sum_{\text{$j$: type $i$ job in the system at $t^*$}}c_i(a(j))} = r_i(A_i(t^*)).\]
    Note that $A_i$ has the same distribution as $T_i$ due to the coupling sample path argument above. Therefore we have that the expected cost in the two problems are the same.
\end{proof}





\section{Deferred proofs in Section~\ref{sec:simpli}}
\label{app:defer proofs}

\begin{lemma}
\label{lemma: app for cost bar}
    When $t\to \infty$, $\lim_{t\to \infty}\costbar(t,\alpha)e^{-\frac{\alpha}{1-\gamma_1}t} = 0$ and $\lim_{t\to \infty}\E{r(t+X)}e^{-\frac{\alpha}{1-\gamma_1}t}  = 0$.
\end{lemma}

\begin{proof}
    \begin{eqnarray*}
        \lim_{t\to \infty}\costbar(t,\alpha)e^{-\frac{\alpha}{1-\gamma_1}t}
        &= &\lim_{t\to \infty} \E{\int_0^{T_1}\alpha e^{-\alpha s} r(t+A(s))ds} e^{-\frac{\alpha}{1-\gamma_1}t}\\
        &\leq & \lim_{t\to \infty} \E{\int_0^{T_1}\alpha  r(t+A(s))ds} e^{-\frac{\alpha}{1-\gamma_1}t}\\
        &= &\lim_{t\to \infty} \alpha \frac{1}{\mu-\lambda} \E{r(t+A)}e^{-\frac{\alpha}{1-\gamma_1}t} \qquad  \text{where } A\sim \Exp(\mu-\lambda)\footnotemark\\
        &\stackrel{(*)}{\leq} &  \lim_{t\to \infty} \frac{\alpha }{\mu-\lambda} \E{r(t+A)}e^{-(\mu-\lambda)t}\\
        & = & 0.
    \end{eqnarray*}
\footnotetext{$A$ denotes the age of the oldest job in an M/M/1 queue, and $A\sim \Exp(\mu-\lambda)$ given $A>0$. }
    \begin{eqnarray*}
        \lim_{t\to \infty}\E{r(t+X)}e^{-\frac{\alpha}{1-\gamma_1}t} 
        &= &\lim_{t\to \infty}\int_{0}^\infty r(t+x) e^{-\frac{\alpha}{1-\gamma_1}t} \frac{\alpha}{1-\gamma_1}e^{\frac{\alpha}{1-\gamma_1} x}dx \\
        &= &\lim_{t\to \infty}\int_{0}^\infty r(t+x) e^{-\frac{\alpha}{1-\gamma_1}(t+x)} \frac{\alpha}{1-\gamma_1}e^{\frac{\alpha}{1-\gamma_1} 2x}dx \\
        &\stackrel{(*)}{=} &\lim_{t\to \infty} \frac{1}{2}\int_{0}^\infty r(t+x) e^{-(\mu-\lambda)(t+x)} \frac{2\alpha}{1-\gamma_1}e^{\frac{\alpha}{1-\gamma_1} 2x}dx\\
        &= &\lim_{t\to \infty} e^{-(\mu-\lambda)t}\frac{1}{2}\int_{0}^\infty r(t+x) e^{-(\mu-\lambda)x} \frac{2\alpha}{1-\gamma_1}e^{\frac{\alpha}{1-\gamma_1} 2x}dx  \\
        &= &0. 
    \end{eqnarray*}
    where the steps $(*)$ are obtained by using 
    \[\frac{\alpha}{1-\gamma_1} = \frac{\alpha}{1-\E{e^{-\alpha T_1}}}\geq \frac{1}{\E{T_1}}=\mu-\lambda,\]
    and the inequality is obtained by $1-e^{-\alpha x}\leq \alpha x$.
\end{proof}

\begin{lemma}
\label{lemma: four eq app}
Let $T_2\sim \Exp(\lambda)$, $X\sim \Exp(\frac{\alpha}{1-\gamma_1})$. Then we have the following equations:
\begin{align}
    &\E{r'(t-T_2)}=\lambda c(t).\\
    &\E{r(t-T_2)}= r(t)-c(t).\\
    &\E{r'(t+X)}= \frac{\mu}{\gamma_1}\E{c(t+X)} - \frac{\alpha}{1-\gamma_1}c(t).\\
    &\E{r(t+X)}=r(t)-c(t)+\frac{\mu(1-\gamma_1)}{\gamma_1\alpha}\E{c(t+X)}.
\end{align}

\end{lemma}
\begin{proof}
    We prove the four equations by using Lemma~\ref{lemma:exp formula} and Lemma~\ref{lemma: r'}:
\begin{align}
    \E{r'(t-T_2)} &= \E{c'(t-T_2)} + \lambda\E{c(t-T_2)}  & Lemma~\ref{lemma: r'} \nonumber\\
    &= \lambda (c(t)-\E{c(t-T_2)})+ \lambda\E{c(t-T_2)}  & Lemma~\ref{lemma:exp formula} \nonumber\\
    &= \lambda c(t).\nonumber
\end{align}

\begin{equation*}
    \E{r(t-T_2)} = r(t)-\frac{1}{\lambda}\E{r'(t-T_2)} = r(t)-c(t). \qquad \qquad \qquad \qquad \qquad \qquad
\end{equation*}

\begin{align}
    \E{r'(t+X)}&= \E{c'(t+X)} + \lambda\E{c(t+X)}  & Lemma~\ref{lemma: r'} \nonumber\\
    &= \frac{\alpha}{1-\gamma_1}\left(\E{c(t+X)} -c(t)\right) + \lambda\E{c(t+X)}& Lemma~\ref{lemma:exp formula} \nonumber\\
    &= \frac{\mu}{\gamma_1}\E{c(t+X)} - \frac{\alpha}{1-\gamma_1}c(t) & Equation~\eqref{eq:gamma2}.\nonumber
\end{align}

\begin{align}
    \E{r(t+X)} &= r(t)+\frac{1-\gamma_1}{\alpha}\E{r'(t+X)} = r(t)-c(t)+\frac{\mu(1-\gamma_1)}{\gamma_1\alpha}\E{c(t+X)}.\nonumber
\end{align}
\end{proof}

\section{Proof for indexability}
\label{app:index}
The goal of this appendix section is to prove Theorem~\ref{thm:index app}. We need several lemmas for the proof.

We first characterize the optimal policy. Specifically, we show that the policy $Threshold(t_0)$ is optimal when the compensation $\ell=\guess(t_0,\alpha)$. This is proved by showing the Hamilton–Jacobi–Bellman equation is satisfied.

\begin{lemma}
\label{lemma: main in index}
    The optimal policy under compensation $\guess(t_0,\alpha)$ is the policy $Threshold(t_0)$.
\end{lemma}
\begin{proof}
    It suffices to prove that the policy $Threshold(t_0)$ satisfies the Hamilton–Jacobi–Bellman equation. Let $V(t)$ denote the value function of the policy at state $t$. 

    At state $t$, if during the next $\delta$ time the action is passive, the total cost is 
    \begin{equation}
    \label{eq:passive}
        passive(t):=(\alpha r(t)-\ell)\delta + e^{-\alpha \delta}V(t+\delta).
    \end{equation}
    Moreover, if the action is active, the total cost is
    \begin{equation}
    \label{eq:active}
        active(t):=(\alpha r(t))\delta + e^{-\alpha\delta}\left( (1-\mu\delta) V(t+\delta)+\mu\delta \E{V(t+\delta-T_2)}\right),
    \end{equation}
    where $T_2\sim \Exp(\lambda)$. 
    Thus the Hamilton–Jacobi–Bellman equation is equivalent to 
    \begin{align}
    t\leq t_0 \Rightarrow \lim_{\delta\to 0}\frac{passive(t)-active(t)}{\delta} \leq 0.
        \label{eq:Bellman orig}\\
    t\geq t_0 \Rightarrow \lim_{\delta\to 0}\frac{passive(t)-active(t)}{\delta} \geq 0.
        \label{eq:Bellman orig2}
    \end{align}

    By \eqref{eq:passive} and \eqref{eq:active}, it is further equivalent to 
    \begin{align}
        t\leq t_0 \Rightarrow \ell\geq \mu\left(V(t)-\E{V(t-T_2)} \right).
        \label{eq:bellman}\\
        t\geq t_0 \Rightarrow \ell\leq \mu\left(V(t)-\E{V(t-T_2)} \right).
        \label{eq:bellman2}
    \end{align}

    \textbf{Case 1: $t\leq t_0$.} In this case, we want to prove \eqref{eq:bellman}.

    Note that the policy stays passive when the state is smaller then $t_0$.  Thus we have that 

    \begin{align}
        V(t)&= \int_0^{t_0-t} (\alpha r(t+s) - \ell) e^{-\alpha s} ds + e^{-\alpha (t_0-t)} V(t_0)\nonumber\\
        &= \int_0^{t_0-t} \alpha r(t+s) e^{-\alpha s} ds + e^{-\alpha (t_0-t)} V(t_0) - (1-e^{-\alpha(t_0-t)})\frac{\ell}{\alpha}.
        \label{eq:Case1 Vt}
    \end{align}

    \begin{align}
        \E{V(t-T_2)} &= \E{\int_{0}^{T_2} (\alpha r(t-T_2+ s)-\ell) e^{-\alpha s}ds} + \E{e^{-\alpha T_2}} V(t) \nonumber\\
        &= \Gamma(t,\alpha) - \frac{\ell}{\alpha}(1-\gamma_2) + \gamma_2 V(t).
        \label{eq:Evt-T2}
    \end{align}

    Thus, using \eqref{eq:guess W}, \eqref{eq:Evt-T2},\eqref{eq:Case1 Vt}, we have that \eqref{eq:bellman} is equivalent to 
    \begin{equation}
        \E{c(t_0+X)}\geq (1-\gamma_2)\left(\int_0^{t_0-t} \alpha r(t+s) e^{-\alpha s} ds + e^{-\alpha(t_0-t)}(V(t_0)+\frac{\ell}{\alpha})  \right) - \Gamma(t,\alpha).
        \label{eq:bellman4}
    \end{equation}

    Using \eqref{eq:cost expr}, \eqref{eq: equation1}, \eqref{eq: equation2}, \eqref{eq: equation4}, \eqref{eq: equation5} and \eqref{eq:guess W}, we have that
    \begin{align}
        V(t_0) =& \frac{1}{1-\gamma_1\gamma_2}\left(\costbar(t_0,\alpha) + \gamma_1\Gamma(t_0,\alpha) - \frac{\ell}{\alpha}\gamma_1(1-\gamma_2)\right)  & \eqref{eq:cost expr}\nonumber\\
        =&\frac{1}{1-\gamma_1\gamma_2}\bigg((1-\gamma_1)\E{r(t_0+X)} + \gamma_1(1-\gamma_2)\E{r(t_0-T_2)} - \frac{\ell}{\alpha}\gamma_1(1-\gamma_2)\bigg) & \text{by }\eqref{eq: equation1}, \eqref{eq: equation2} \nonumber\\
        =&\frac{1}{1-\gamma_1\gamma_2}\bigg((1-\gamma_1)\left(r(t_0)-c(t_0)+\frac{\mu(1-\gamma_1)}{\gamma_1\alpha}\E{c(t_0+X)}\right) &\text{by }\eqref{eq: equation6} \nonumber\\
        &+ \gamma_1(1-\gamma_2)(r(t_0)-c(t_0)) - \frac{\ell}{\alpha}\gamma_1(1-\gamma_2)\bigg) &\text{by }\eqref{eq: equation4}\nonumber\\
        =&\frac{1}{1-\gamma_1\gamma_2}\bigg((1-\gamma_1)\left(r(t_0)-c(t_0)+\frac{1-\gamma_1}{\gamma_1\alpha}\ell\right) &\text{by }\eqref{eq:guess W} \nonumber\\
        &+ \gamma_1(1-\gamma_2)(r(t_0)-c(t_0)) - \frac{\ell}{\alpha}\gamma_1(1-\gamma_2)\bigg) \nonumber\\ 
        =& r(t_0)-c(t_0) + \frac{1}{1-\gamma_1\gamma_2}\left(\frac{(1-\gamma_1)^2}{\gamma_1} - \gamma_1(1-\gamma_2)\right)\frac{\ell}{\alpha}\nonumber\\
        =& r(t_0)-c(t_0) -\frac{\ell}{\alpha}+\frac{1-\gamma_1}{\gamma_1(1-\gamma_1\gamma_2)}\frac{\ell}{\alpha} \nonumber\\
        =& r(t_0)-c(t_0) -\frac{\ell}{\alpha} +\frac{\lambda+\alpha}{\mu}\frac{\ell}{\alpha}. &\text{by }\eqref{eq:gamma2}
        \label{eq:V(t0)}
    \end{align}

    Thus again using \eqref{eq:gamma2}, we have that \eqref{eq:bellman4} is equivalent to
    \[\E{c(t_0+X)}\geq (1-\gamma_2)\left(\int_0^{t_0-t} \alpha r(t+s)  e^{-\alpha s} ds + e^{-\alpha(t_0-t)}\left(r(t_0)-c(t_0) + \frac{1}{1-\gamma_2}\frac{\ell}{\mu}\right)\right) - \Gamma(t,\alpha).\]

    Using \eqref{eq: equation2}, \eqref{eq: equation4}, \eqref{eq:gamma2} and reordering the terms, the inequality is further equivalent to 
    \begin{align}
        &\left(1-e^{-\alpha(t_0-t)}\right)\E{c(t_0+X)} + (1-\gamma_2)\left(-e^{-\alpha(t_0-t)}(r(t_0)-c(t_0)) + r(t)-c(t)\right)\nonumber\\&\geq (1-\gamma_2)\int_0^{t_0-t} \alpha r(t+s) e^{-\alpha s}ds .
        \label{eq:bellman 5}
    \end{align}

    Note that the right hand side of \eqref{eq:bellman 5} can be simplified to:
    \begin{align*}
        &(1-\gamma_2)\int_0^{t_0-t} \alpha r(t+s) e^{-\alpha s}ds \\
        &= (1-\gamma_2) \left( -r(t+s)e^{-\alpha s}\bigg|_0^{t_0-t} + \int_0^{t_0-t} r'(t+s)e^{-\alpha s} ds\right) \\
        &= (1-\gamma_2) \left( r(t) - e^{-\alpha (t_0-t)} r(t_0) + \int_0^{t_0-t} (c'(t+s)+\lambda c(t+s))e^{-\alpha s} ds\right) & Lemma~\ref{lemma: r'}\\
        &\leq (1-\gamma_2) \left( r(t) - e^{-\alpha (t_0-t)} r(t_0) + \int_0^{t_0-t} c'(t+s) ds + \int_0^{t_0-t} \lambda c(t_0)e^{-\alpha s} ds\right)\\
        &= (1-\gamma_2) \left( r(t) - e^{-\alpha (t_0-t)} r(t_0) +c(t_0)-c(t) + \frac{\lambda}{\alpha} (1-e^{-\alpha (t_0-t)}) c(t_0)\right).
    \end{align*}

    Thus to prove \eqref{eq:bellman 5}, we only need to prove that 
    \begin{align}
        &\frac{1}{1-\gamma_2}\left(1-e^{-\alpha(t_0-t)}\right)\E{c(t_0+X)} + \left(-e^{-\alpha(t_0-t)}(r(t_0)-c(t_0)) + r(t)-c(t)\right)\nonumber\\&\geq
        r(t) - e^{-\alpha (t_0-t)} r(t_0) +c(t_0)-c(t) + \frac{\lambda}{\alpha} (1-e^{-\alpha (t_0-t)}) c(t_0),\nonumber
    \end{align}

    which is equivalent to 
    \[\E{c(t_0+X)}\geq c(t_0).\]

    This holds by monotonocity of $c$.
    
    \textbf{Case 2: $t>t_0$.} In this case, we want to prove
    \begin{equation}
        \ell\leq \mu\left(V(t)-\E{V(t-T_2)} \right).
        \label{eq:obj in case2}
    \end{equation}
    Note that starting from state $t$, the policy stays active until the state drops back to $t$, during which time a $\costbar(t,\alpha)$ cost is incurred. After the state drops below $t$, since each time the state drops by an exponential amount, the state is $t-T_2$. Thus we have that 
    \begin{equation}
        \label{eq:Vt and Vt-T2}
        V(t) = \costbar(t,\alpha) + \gamma_1 \E{V(t-T_2)}.
    \end{equation}

    Substituting \eqref{eq:Vt and Vt-T2} into \eqref{eq:obj in case2} and using \eqref{eq:guess W}, we only need to prove that 
    \begin{equation}
        \E{c(t_0+X)}\leq  -\frac{1-\gamma_1}{\gamma_1}V(t)+\frac{1}{\gamma_1}\costbar(t,\alpha).
        \label{eq:obj simp}
    \end{equation}

    Now we characterize $V(t)$. Starting from state $t$, the policy stays active until the state drops below $t$. During this time ($T_1$, which follows the Busy period in M/M/1 distribution), a cost of $\costbar(t,\alpha)$ is incurred. 
    After the state drops below $t$, the amount it is below $t$ follows on exponential distribution with rate $\lambda$. Suppose it is $y_1\sim t-\Exp(\lambda)$. If $y_1$ is still larger than $t_0$, the policy stays active until the state drops below $y_1$, incurring another $\costbar(y_1,\alpha)$ cost. This process continues until the state drops below $t_0$. Note that each time a $\costbar(y_i,\alpha)$ is incurred, the state drops by an exponential amount, thus the total number of iterations follow a Poisson distribution $Pois(\lambda(t-t_0))$. After the final iteration, the state is $\Exp(\lambda)$ below $t_0$, and the policy stays passive until the state grows back to $t_0$.
    
    Thus we have that 
    \begin{equation}
        V(t) = \costbar(t,\alpha) + \gamma_1 \E{\sum_{i=1}^M \gamma_1^{i-1}\costbar(y_i,\alpha)} + \E{\gamma_1^{M+1}}\left(\Gamma(t_0,\alpha)-\frac{\ell}{\alpha}(1-\gamma_2)+\gamma_2 V(t_0)\right),
        \label{eq:Vt1}
    \end{equation}
    where $y_1=t-\Exp(\lambda),y_{i+1}=y_i-\Exp(\lambda)$, and $M$ is the random variable such that $y_M\geq t_0,y_{M+1}<t_0$, with $M\sim Pois(\lambda(t-t_0))$. By the probability generating function of a Poisson variable, we have that 
    \begin{equation}
        \E{\gamma_1^{M+1}}=\gamma_1 PGF(\gamma_1)=\gamma_1 e^{-\lambda(t-t_0)(1-\gamma_1)}.
        \label{eq:PGF}
    \end{equation}

    Define $\PoisCost(t,t_0):= \E{\sum_{i=1}^M \gamma_1^{i-1}\costbar(y_i,\alpha)}$. We treat $t_0$ as constant and vary $t$ by $\delta$. Note that $\{y_i\}$ is distributed as Poisson arrivals in time $(t_0,t)$, we discuss different cases based on the number of $y_i$ in the interval $(t,t+\delta)$, denoted by $N$. Note that $N$ follows the distribution $Pois(\lambda \delta)$. 
    
    If $N=0$ (with probability $e^{-\lambda\delta}$), $\PoisCost(t+\delta,t_0)=\PoisCost(t,t_0)$. 
    
    If $N=1$ (with probability $\lambda\delta e^{-\lambda\delta}$), we have that 
    \[\PoisCost(t+\delta,t_0) = \costbar(t+\Delta,t_0) + \gamma_1 \PoisCost(t,t_0),\]
    where $0<\Delta<\delta$.

    The probability that  $N>1$ is $O(\delta^2)$.

    Thus we have that 
    \begin{equation}
        \PoisCost(t+\delta,t_0) = e^{-\lambda\delta}\PoisCost(t,t_0) + \lambda\delta e^{-\lambda\delta}\left(\costbar(t+\Delta,\alpha) + \gamma_1 \PoisCost(t,t_0)\right) + O(\delta^2).
    \end{equation}

    Therefore,
    \begin{align}
        \PoisCost'(t,t_0) &:= \lim_{\delta\to 0}\frac{\PoisCost(t+\delta,t_0) - \PoisCost(t,t_0)}{\delta}\nonumber\\
        &=\lim_{\delta\to 0} \frac{e^{-\lambda\delta} + \lambda\delta e^{-\lambda\delta}\gamma_1 - 1}{\delta}\PoisCost(t,t_0) + \lambda e^{-\lambda\delta}\costbar(t+\Delta,\alpha)\nonumber\\
        &= (\lambda\gamma_1-\lambda)\PoisCost(t,t_0)+\lambda \costbar(t,\alpha).
    \end{align}

    Note that this is a first-order ODE with respect to $t$. Using Lemma~\ref{lemma:ODE} and treating $y=\PoisCost(t,t_0),P(t) = \lambda(1-\gamma_1),Q(t)=\lambda \costbar(t,\alpha)$, we have that 
    \[\PoisCost(t,t_0) = e^{-\lambda(1-\gamma_1)t}\left( 
     \int e^{\lambda(1-\gamma_1)t}\lambda \costbar(t,\alpha) dt + C\right).\]
     Noting that $\PoisCost(t_0,t_0) = 0$, we have that 
     \begin{align}
         \PoisCost(t,t_0) &= e^{-\lambda(1-\gamma_1)t}\int_{t_0}^t e^{\lambda(1-\gamma_1)s}\lambda \costbar(s,\alpha) ds \nonumber\\
         &= \int_{t_0}^t e^{-\lambda(1-\gamma_1)(t-s)}\lambda \costbar(s,\alpha) ds.
         \label{eq:PoisCost}
     \end{align}

    Now substituting \eqref{eq:PoisCost} and \eqref{eq:PGF} into \eqref{eq:Vt1}, we have that 
    \begin{align}
        V(t) =& \costbar(t,\alpha) + \gamma_1 \PoisCost(t,t_0) + \E{\gamma_1^{M+1}}\left(\Gamma(t_0,\alpha)-\frac{\ell}{\alpha}(1-\gamma_2)+\gamma_2 V(t_0)\right)\nonumber\\
        =&\costbar(t,\alpha) + \gamma_1 \int_{t_0}^t e^{-\lambda(1-\gamma_1)(t-s)}\lambda \costbar(s,\alpha) ds\nonumber \\
        &+ 
        \gamma_1 e^{-\lambda(t-t_0)(1-\gamma_1)}\left(\Gamma(t_0,\alpha)-\frac{\ell}{\alpha}(1-\gamma_2)+\gamma_2 V(t_0)\right).
        \label{eq:Vt2}
    \end{align}    

    Moreover, we have that 
    \begin{align}
        \int_{t_0}^t e^{-\lambda(1-\gamma_1)(t-s)}\lambda \costbar(s,\alpha) ds =& \frac{1}{1-\gamma_1}e^{-\lambda(1-\gamma_1)(t-s)} \costbar(s,\alpha) \bigg |_{t_0}^t\nonumber \\
        &- \int_{t_0}^t \frac{1}{1-\gamma_1}e^{-\lambda(1-\gamma_1)(t-s)} \costbar'(s,\alpha)ds\nonumber \\
        &= \frac{1}{1-\gamma_1}\left(\costbar(t,\alpha) - e^{-\lambda(1-\gamma_1)(t-t_0)}\costbar(t_0,\alpha)\right)\nonumber\\
        &- \int_{t_0}^t \frac{1}{1-\gamma_1}e^{-\lambda(1-\gamma_1)(t-s)} \costbar'(s,\alpha)ds. 
        \label{eq:int term}
    \end{align}

    By \eqref{eq:gamma2}, \eqref{eq: equation2}, \eqref{eq: equation4} and \eqref{eq:V(t0)}, we have that 
    \begin{align}
        &\Gamma(t_0,\alpha)-\frac{\ell}{\alpha}(1-\gamma_2)+\gamma_2 V(t_0)\nonumber \\
        &= (1-\gamma_2)(r(t_0)-c(t_0)) - \frac{\ell}{\alpha}(1-\gamma_2) + \gamma_2\left( r(t_0)-c(t_0) -\frac{\ell}{\alpha} +\frac{\lambda+\alpha}{\mu}\frac{\ell}{\alpha}\right) \nonumber\\
        &=r(t_0)-c(t_0)-\frac{\ell}{\alpha} + \frac{\lambda}{\alpha}\frac{\ell}{\mu}. & \text{by }\eqref{eq:gamma2}
        \label{eq:other term}
        \end{align}

    Thus substituting \eqref{eq:int term} and \eqref{eq:other term} into \eqref{eq:Vt2}, we have that 
    \begin{align}
        V(t) 
        =&\costbar(t,\alpha) + \frac{\gamma_1}{1-\gamma_1}\left(\costbar(t,\alpha) - e^{-\lambda(1-\gamma_1)(t-t_0)}\costbar(t_0,\alpha)\right)\nonumber\\
        &- \int_{t_0}^t \frac{\gamma_1}{1-\gamma_1}e^{-\lambda(1-\gamma_1)(t-s)} \costbar'(s,\alpha)ds\nonumber \\
        &+\gamma_1 e^{-\lambda(t-t_0)(1-\gamma_1)}\left(r(t_0)-c(t_0)-\frac{\ell}{\alpha} + \frac{\lambda}{\alpha}\frac{\ell}{\mu}\right) 
    \end{align}

    Note that by \eqref{eq: equation1}, \eqref{eq: equation6} and \eqref{eq:guess W}, \[\costbar(t_0,\alpha) = (1-\gamma_1)(r(t_0)-c(t_0)+\frac{1-\gamma_1}{\gamma_1}\frac{\ell}{\alpha}).\]

        Thus, we have that 

    \begin{align}
        V(t) 
        =&\frac{1}{1-\gamma_1}\costbar(t,\alpha) - \frac{\gamma_1}{1-\gamma_1}e^{-\lambda(1-\gamma_1)(t-t_0)}\costbar(t_0,\alpha)\nonumber\\
        &- \int_{t_0}^t \frac{\gamma_1}{1-\gamma_1}e^{-\lambda(1-\gamma_1)(t-s)} \costbar'(s,\alpha)ds\nonumber \\
        &+\gamma_1 e^{-\lambda(t-t_0)(1-\gamma_1)}\left(r(t_0)-c(t_0)-\frac{\ell}{\alpha} + \frac{\lambda}{\alpha}\frac{\ell}{\mu}\right)\nonumber\\
        =&\frac{1}{1-\gamma_1}\costbar(t,\alpha) - \gamma_1e^{-\lambda(1-\gamma_1)(t-t_0)}\left(r(t_0)-c(t_0)+\frac{1-\gamma_1}{\gamma_1}\frac{\ell}{\alpha}\right)\nonumber\\
        &- \int_{t_0}^t \frac{\gamma_1}{1-\gamma_1}e^{-\lambda(1-\gamma_1)(t-s)} \costbar'(s,\alpha)ds\nonumber \\
        &+\gamma_1 e^{-\lambda(t-t_0)(1-\gamma_1)}\left(r(t_0)-c(t_0)-\frac{\ell}{\alpha} + \frac{\lambda}{\alpha}\frac{\ell}{\mu}\right)\nonumber\\
        =& \frac{1}{1-\gamma_1}\costbar(t,\alpha) + \gamma_1e^{-\lambda(1-\gamma_1)(t-t_0)}\left(\frac{\lambda}{\mu}-\frac{1}{\gamma_1}\right)\frac{\ell}{\alpha} \nonumber\\
        &- \int_{t_0}^t \frac{\gamma_1}{1-\gamma_1}e^{-\lambda(1-\gamma_1)(t-s)} \costbar'(s,\alpha)ds 
        \label{eq:Vt3}
    \end{align}

    Now look back to our goal, which is proving \eqref{eq:obj simp}:
    \[\E{c(t_0+X)}\leq  -\frac{1-\gamma_1}{\gamma_1}V(t)+\frac{1}{\gamma_1}\costbar(t,\alpha).\]
    This is equivalent to proving 
    \[V(t)-\frac{1}{1-\gamma_1}\costbar(t,\alpha)\leq -\frac{\gamma_1}{1-\gamma_1}\E{c(t_0+X)}.\]
    By \eqref{eq:Vt3}, our goal is equivalent to proving 
    \begin{equation*}
        \gamma_1e^{-\lambda(1-\gamma_1)(t-t_0)}\left(\frac{\lambda}{\mu}-\frac{1}{\gamma_1}\right)\frac{\ell}{\alpha} - \int_{t_0}^t \frac{\gamma_1}{1-\gamma_1}e^{-\lambda(1-\gamma_1)(t-s)} \costbar'(s,\alpha)ds\leq -\frac{\gamma_1}{1-\gamma_1}\E{c(t_0+X)}
    \end{equation*}
    Rearranging the terms and using \eqref{eq:guess W}, our goal is equivalent to proving that 
    \begin{equation}
        \E{c(t_0+X)}\left(\frac{\gamma_1}{1-\gamma_1}+ e^{-\lambda(1-\gamma_1)(t-t_0)} \left(\frac{\lambda \gamma_1}{\alpha} - \frac{\mu}{\alpha}\right)\right)\leq  \int_{t_0}^t \frac{\gamma_1}{1-\gamma_1}e^{-\lambda(1-\gamma_1)(t-s)} \costbar'(s,\alpha)ds.
        \label{eq:obj3}
    \end{equation}

    Note that by \eqref{eq: equation1} and \eqref{eq: equation6}, we have that 
    \[\costbar(t,\alpha) = (1-\gamma_1)(r(t)-c(t)+\frac{\mu(1-\gamma_1)}{\gamma_1\alpha}\E{c(t+X)}).\]

    Thus using Lemma~\ref{lemma: r'} and \eqref{eq:gamma2},
    \begin{align*}
    \costbar'(t,\alpha) &= (1-\gamma_1)\left(\lambda c(t)+\frac{\mu(1-\gamma_1)}{\gamma_1\alpha}\E{c'(t+X)}\right)    \\
    &= (1-\gamma_1)\left(\lambda c(t)+\frac{\lambda+\alpha-\lambda\gamma_1}{\alpha}\E{c'(t+X)}\right). 
    \end{align*}

    Moreover, by \eqref{eq:gamma2},  
    \[\frac{\lambda \gamma_1}{\alpha} - \frac{\mu}{\alpha} =\frac{1}{\alpha}(\lambda\gamma_1-\frac{(\lambda+\alpha-\lambda\gamma_1)\gamma_1}{1-\gamma_1}) =-\frac{\gamma_1}{1-\gamma_1},\]

    Thus we have that our goal, equivalent to \eqref{eq:obj3}, is equivalent to showing that
    \begin{equation}
        \E{c(t_0+X)}\frac{1}{1-\gamma_1}\left(1-e^{-\lambda(1-\gamma_1)(t-t_0)}\right)\leq \int_{t_0}^t e^{-\lambda(1-\gamma_1)(t-s)} \left(\lambda c(s)+\frac{\lambda+\alpha-\lambda\gamma_1}{\alpha} \E{c'(s+X)}\right)ds.
    \end{equation}

    Finally, we have that for any $s>t_0$,
    \begin{align*}
        &\lambda c(s)+\frac{\lambda+\alpha-\lambda\gamma_1}{\alpha} \E{c'(s+X)} \\
        &= \lambda c(s) + \frac{\lambda+\alpha-\lambda\gamma_1}{\alpha} \cdot \frac{\alpha}{1-\gamma_1}\left(\E{c(s+X)}-c(s)\right) & Lemma~\ref{lemma:exp formula} \\
        &=\left(\lambda - \frac{\lambda+\alpha-\lambda\gamma_1}{1-\gamma_1}\right) c(s) + \frac{\lambda+\alpha-\lambda\gamma_1}{1-\gamma_1} \E{c(s+X)} \\
        &= - \frac{\alpha}{1-\gamma_1} c(s) + \frac{\lambda+\alpha-\lambda\gamma_1}{1-\gamma_1} \E{c(s+X)} \\
        &\geq \left(- \frac{\alpha}{1-\gamma_1} + \frac{\lambda+\alpha-\lambda\gamma_1}{1-\gamma_1}\right) \E{c(s+X)} \\
        &= \lambda \E{c(s+X)}\\
        &\geq \lambda \E{c(t_0+X)}.
    \end{align*}

    Thus 
    \begin{align*}
        &\int_{t_0}^t e^{-\lambda(1-\gamma_1)(t-s)} \left(\lambda c(s)+\frac{\lambda+\alpha-\lambda\gamma_1}{\alpha} \E{c'(s+X)}\right)ds \\
        &\geq \int_{t_0}^t e^{-\lambda(1-\gamma_1)(t-s)} \lambda \E{c(t_0+X)}ds\\
        &= \E{c(t_0+X)}\frac{1}{1-\gamma_1}\left(1-e^{-\lambda(1-\gamma_1)(t-t_0)}\right).
    \end{align*}

\end{proof}

To prove indexability, there are actually two other corner cases to prove. 

\begin{lemma}
    \begin{enumerate}
        \item If the compensation $\ell\leq \mu\E{c(X)}$, always active (except at state 0) is optimal.
        \item If $\lim_{t\to \infty} \mu\E{c(t+X)}=M$ exists  and
        the compensation $\ell\geq \mu M$, always passive is optimal.
    \end{enumerate}
    \label{lemma: corner}
\end{lemma}

\begin{proof}
    To prove the first argument, the argument is the same as what we have in Case 2 of the proof of Lemma~\ref{lemma: main in index}. Now we prove the second argument.

    Since $\lim_{t\to \infty} \mu\E{c(t+X)}=M$ exists and $c$ is monotone, we have that $\lim_{t\to \infty} c(t) = M$.

    For the policy ``always passive", we can derive the value function:
    \begin{align}
        V(t) &= \int_0^\infty (\alpha r(t+s) - \ell) e^{-\alpha s}ds\nonumber\\ 
        &= \int_0^\infty \alpha r(t+s)e^{-\alpha s}ds - \frac{\ell}{\alpha} \nonumber\\
        &= -r(t+s) e^{-\alpha s}\bigg|_0^{\infty} - \int_0^\infty -r'(t+s)e^{-\alpha s} ds - \frac{\ell}{\alpha}\nonumber\\
        &= r(t) + \int_0^\infty (c'(t+s)+\lambda c(t+s))e^{-\alpha s} ds - \frac{\ell}{\alpha}. & Lemma~\ref{lemma: r'}
        \label{eq:passive value}
    \end{align}

    Now we verify that this value function satisfies the Hamilton–Jacobi–Bellman equation, which is equivalent to showing:
    \[\lim_{\delta\to 0} \frac{passive(t)-active(t)}{\delta} \leq 0.\]
    By the same argument in \eqref{eq:passive} and \eqref{eq:active}, our goal is equivalent to showing
    \begin{equation}
        \ell \geq \mu(V(t) - \E{V(t-T_2)}).
    \end{equation}

    Note that \eqref{eq:Evt-T2} still holds. Thus it suffices to prove that 
    \[\ell\geq \mu \left((1-\gamma_2) (V(t)+\frac{\ell}{\alpha}) - \Gamma(t,\alpha)\right).\]
    Using \eqref{eq:passive value}, \eqref{eq: equation2} and \eqref{eq: equation4}, we only need to prove that 
    \[\frac{\ell}{\mu}\geq (1-\gamma_2)\left(r(t) + \int_0^\infty (c'(t+s)+\lambda c(t+s))e^{-\alpha s} ds -(r(t)-c(t))\right).\]

    Note that the right hand side can be bounded as follows:
    \begin{align*}
        &(1-\gamma_2)\left(r(t) + \int_0^\infty (c'(t+s)+\lambda c(t+s))e^{-\alpha s} ds -(r(t)-c(t))\right) \\
        &\leq (1-\gamma_2)\left(c(t) + \int_0^\infty c'(t+s) ds + \int_0^\infty \lambda M e^{-\alpha s} ds \right)\\
        &= (1-\gamma_2)(\lim_{t\to \infty} c(t) + \frac{\lambda}{\alpha} M )\\
        &= (1-\gamma_2) \frac{\lambda+\alpha}{\alpha} M\\
        &= M.
    \end{align*}
    Since $\ell\geq \mu M$, we have the proof.
\end{proof}

Finally we can prove indexability. 

\begin{theorem}[Indexability]
\label{thm:index app}
    The discounted single-arm bandit problem is indexable with the Whittle Index $W(t_0,\alpha)=\guess(t_0,\alpha)$.
\end{theorem}

\begin{proof}
We first prove indexability.
Recall that \[\Pi(\ell):=\{t_0 \mid \text{passive action at state $t_0$ is optimal}\}.\]
    Now we prove that for any $\ell_1<\ell_2$, we have that $\Pi(\ell_1)\subseteq \Pi(\ell_2)$.

    If $\ell_1\leq \mu\E{c(X)}$, by Lemma~\ref{lemma: corner}, $\Pi(\ell_1)=\{0\}$, and the statement holds immediately. 

    If $\ell_2 \geq \lim_{t\to \infty} \mu\E{c(t+X)}$, by Lemma~\ref{lemma: corner}, $\Pi(\ell_2)=\mathbb{R}$, and the statement holds immediately. 

    Otherwise, we have that for some $t_1,t_2$, $\ell_1 = \guess(t_1,\alpha)$, $\ell_2 = \guess(t_2,\alpha)$. Since $\guess(t,\alpha)$ is non-decreasing with $t$, we have that $t_1<t_2$. By Lemma~\ref{lemma: main in index}, $\Pi(\ell_1) = \{t\mid t\leq t_1\}$, and $\Pi(\ell_2) = \{t\mid t\leq t_2\}$. Thus we prove the statement. 

    The proof that the Whittle Index is $\guess(t,\alpha)$ is straightforward. Mathematically, we want to show that 
    \[\guess(t_0,\alpha) = \inf_\ell\{t\in \Pi(\ell)\}.\]
    This follows immediately by Lemma~\ref{lemma: corner} and Lemma~\ref{lemma: main in index}.
\end{proof}

\end{document}